\documentclass[a4paper,11pt]{article}
\usepackage[utf8]{inputenc}
\usepackage{amsmath}
\usepackage{graphicx}

\usepackage{amsmath}
\usepackage{amssymb}
\usepackage{ascmac}
\usepackage{amsthm}
\usepackage{here}
\theoremstyle{plain}
\newtheorem{theorem}{Theorem}
\newtheorem{lemma}[theorem]{Lemma}

\theoremstyle{definition}
\newtheorem{definition}[theorem]{Definition}
\newtheorem{example}[theorem]{Example}

\newtheorem{assumption}{Assumption}

\usepackage{algorithm}
\usepackage[noend]{algorithmic}
\usepackage{bm}
\usepackage{url}
\usepackage{cases}
\usepackage{graphicx}
\usepackage{comment}
\usepackage{blindtext}

\usepackage{booktabs}
\usepackage{mleftright}\mleftright
\usepackage{multicol}
\usepackage{authblk}
\usepackage{enumerate}
\usepackage{bbm}

\usepackage{thm-restate}
\usepackage{thmtools}

\usepackage{enumitem}

\usepackage{bm}
\usepackage{authblk}

\usepackage[margin=2.3cm]{geometry}
\usepackage{microtype}  
\usepackage{hyphenat}

\allowdisplaybreaks

\usepackage[hidelinks]{hyperref}
\usepackage[capitalize,nameinlink]{cleveref}

\usepackage{tikz}
\usepackage{tikz-cd}
\usetikzlibrary{positioning,chains,fit,shapes,decorations.pathreplacing,calligraphy,matrix,backgrounds}

\usepackage{hhline}
\usepackage{multirow}
\newcommand*\circled[1]{\tikz[baseline=(char.base)]{
            \node[shape=circle,draw,inner sep=1pt] (char) {#1};}}
\usetikzlibrary{decorations.pathreplacing}
\allowdisplaybreaks

\usepackage{threeparttable}
\usepackage{color}
\usepackage{CJKutf8}
\usepackage{setspace}

\usepackage{color}
\usepackage{graphicx,colortbl,multirow}
\usepackage{subcaption}

\definecolor{myblue}{RGB}{80,80,160}
\definecolor{mygreen}{RGB}{80,160,80}



\usepackage[color={red!100!green!33},textsize=scriptsize]{todonotes}

\usepackage[right]{lineno}

\def\final{0}  
\def\iflong{\iffalse}
\ifnum\final=0  

\else 
\newcommand{\rnote}[1]{}

\fi

\begin{document}

\title{Asymptotic Existence of Class Envy-free Matchings}
\author[1]{Tomohiko Yokoyama}
\author[1]{Ayumi Igarashi}
\affil[1]{The University of Tokyo}
\affil[ ]{\texttt{ \{\href{mailto:tomohiko_yokoyama@mist.i.u-tokyo.ac.jp}{tomohiko\_yokoyama},\href{mailto:igarashi@mist.i.u-tokyo.ac.jp}{igarashi}\}@mist.i.u-tokyo.ac.jp}}

\date{}

\maketitle


\begin{abstract}
We consider a one-sided matching problem where agents who are partitioned into disjoint classes and each class must receive fair treatment in a desired matching. This model, proposed by Benabbou et al.~\cite{Benabbou2019}, aims to address various real-life scenarios, such as the allocation of public housing and medical resources across different ethnic, age, and other demographic groups. Our focus is on achieving class envy-free matchings, where each class receives a total utility at least as large as the maximum value of a matching they would achieve from the items matched to another class. While class envy-freeness for worst-case utilities is unattainable without leaving some valuable items unmatched, such extreme cases may rarely occur in practice. To analyze the existence of a class envy-free matching in practice, we study a distributional model where agents' utilities for items are drawn from a probability distribution. Our main result establishes the asymptotic existence of a desired matching, showing that a round-robin algorithm produces a class envy-free matching as the number of agents approaches infinity. 
\end{abstract}

\section{Introduction}
One-sided matching is a fundamental problem that forms the economic foundation of numerous practical applications, spanning from kidney exchange~\cite{kidneyexchange}, assigning drivers to customers~\cite{banerjee2019ride}, to house allocation~\cite{Sonmez1998,Hylland1979,ZHOU1990}. An instance of the one-sided matching problem consists of a set of agents, a set of indivisible items, and preferences of the agents over the items. The goal is to find an assignment of items to agents while ensuring desirable normative properties. In particular, guaranteeing fairness is of paramount importance in scenarios such as allocating tasks to workers or providing social housing to residents. 

A substantial body of the literature is dedicated to ensuring fairness among individuals, focusing on concepts such as \emph{envy-freeness}~\cite{BOGOMOLNAIA2001}. 
Envy-freeness ensures that no agent prefers another agent's allocated item to their own. 
On the other hand, when addressing various real-life resource allocation problems, it becomes increasingly crucial to consider fairness requirements among different classes of agents, particularly when implementing solutions in large-scale systems. 
For instance, consider the allocation of scarce medical resources to different regions in a country. Another motivating example involves student placement in public schools among different ethnic groups and social housing allocation among different income groups.
To gain public acceptance, the social planner must ensure an equitable allocation of resources across various regional groups or classes. 

The scenarios we consider are effectively captured by a model proposed by Benabbou et al.~\cite{Benabbou2019}. This model consists of $m$ items and $n$ agents, divided into $k$ disjoint classes, with each agent allocated at most one item. A key feature of this framework is the use of \emph{assignment valuations} to evaluate {\em envy} among classes. Assignment valuations quantify the potential value a class could derive from the set of items allocated to another class, determined by the optimal matching between the items and the members of the class. By adapting fairness notions from fair division literature into the one-sided matching problem, Benabbou et al.~\cite{Benabbou2019} introduced the notion of \emph{class envy-freeness}~\cite{Benabbou2019,HosseiniHuangIgarashiShah2022}, which requires that no class prefers the set of items assigned to any other class over its own bundle, in the sense of assignment valuations.

Unfortunately, a class envy-free matching is not guaranteed to exist without wasting any items. 
Consider a simple example with one item and two single-agent classes, both valuing the item. In any nonempty matching, one class receives nothing whereas the other class receives one item. 
However, this strong conflict between fairness and efficiency does not preclude the existence of practical solutions that balance fairness and efficiency. 
The question then becomes: can we achieve a fair and efficient matching in practical scenarios? 

In the context of fair division, related works have examined the existence of envy-free allocations in probabilistic settings, where agents' utilities for items are modeled using probability distributions~\cite{Amanatidis2017,BaiGolz2022,BaiEC2022,BenadeHalpernPsomasVerma2023,Dickerson2014,Kurokawa2016,ManurangsiSuksompong2017,ManurangsiSuksompong2020,ManurangsiSuksompong2021,Suksompong2016}. 
These approaches of asymptotic analysis are not limited to the fair division but have been widely studied in broader economic models; examples include one-to-one house allocation~\cite{GanSuksompongVoudouris2019,ManurangsiSuksompong2021}, two-sided matchings~\cite{kojima2009incentives,KPR13a,Ashlagi2014}, and voting theory~\cite{slinko2002asymptotic,xia2020smoothed}. 
These studies provide insights into the behavior of mechanisms in large-scale scenarios, and offer both theoretical foundations and practical implications for real-world applications.

\paragraph{Our contributions.}
In this paper, to investigate the asymptotic existence of a matching that is efficient and fair among classes on average, 
we introduce a distributional model where each agent's utilities are drawn from a probability distribution.

We first prove that when $m$ increases quadratically with respect to $n$, a welfare-maximizing mechanism, which returns a maximum-weight matching between agents and items, is asymptotically class envy-free. Importantly, this matching is also efficient since it maximizes the total utility of the agents.
While this result provides asymptotic fairness and efficiency guarantees, it may not be satisfactory in all scenarios. 
Specifically, the mechanism may fail to treat classes fairly for worst-case inputs; for example, the welfare-maximizing mechanism could allocate all items to one class if the members of that class assign large values to each item.

To address this limitation, we integrate the \emph{round-robin algorithm} from fair division into our framework. Specifically, in each round of the algorithm, each class selects an item with the highest marginal utility among the remaining items, creating a maximum-weight matching.  
Our main result establishes that this algorithm asymptotically yields a class envy-free matching as the number of agents approaches infinity, contingent upon milder assumptions regarding the number and sizes of classes.
Furthermore, based on the recent result of \cite{Amanatidis2023}, the round-robin algorithm is known to produce a matching that satisfies an approximation of \emph{class envy-freeness up to one item}  (class EF1) for every input. We also prove that it achieves \emph{non-wastefulness}, satisfying the property that no item can be reallocated to increase one class's valuation without decreasing another's.

While our algorithm shares similarities with the round-robin algorithm designed for fair division instances with additive agents, the non-additivity of valuations introduces notable distinctions in our analysis compared to the additive setting~\cite{ManurangsiSuksompong2021}. This leads to unique technical challenges. Specifically, analyzing the behavior of the round-robin algorithm becomes complex because the set of items available at each round depends on the selections made in prior rounds.

In the case of additive valuations, Manurangsi and Suksompong~\cite{ManurangsiSuksompong2021} utilized a key property: envy between any pair of agents, based on the output of the round-robin algorithm, can be decomposed into the maximum value of an item obtained by another agent and the differences in value between the items each agent received and those received by another agent in consecutive rounds. Exploiting this fact, they showed that these differences can ``catch up to'' the maximum single-item value when there are sufficiently many items.
However, achieving such a decomposition becomes challenging for assignment valuations due to the combinatorial structure of the matchings. In our setting, the marginal contribution of an item to different bundles may vary, and the domination property that an item chosen in an earlier round has a value greater than or equal to an item chosen later no longer holds. Consequently, unlike the additive case, the round-robin may not yield a class EF1 matching (see Example 1 in \cite{HosseiniHuangIgarashiShah2022}).

To overcome these challenges, our proof critically leverages techniques from random assignment theory~\cite{Aldous2001, BuckChanRobbins2002, FriezeJohansson2017, Talagrand1995, Wastlund2005}. This theory considers a bipartite graph with random edge weights, primarily focusing on analyzing the expected value of a minimum weight perfect matching (which is essentially equivalent to a maximum-weight perfect matching). 

In proofs, instead of examining individual pairs of items allocated to two classes, $p$ and $q$, in each round, we focus on the marginal utility that class $p$ receives in each round. This approach helps us analyze the expected value of the items allocated to class $p$. First, we derive a lower bound on the expected total utility for class $p$ by applying the novel techniques introduced by \cite{Wastlund2005, Wastlund2009} and \cite{FriezeJohansson2017}, which involve introducing a special vertex whose behavior explicitly determines the expected marginal gain. Next, we evaluate the expected value of the items allocated to class $q$ from class $p$'s perspective, using a randomly selected bundle of the same size as class $p$'s. Finally, we discuss the concentration of probability around these expected values, showing that the edge weights chosen by the round-robin algorithm are sufficiently large. 
We discuss our proof techniques in Sections~\ref{sec:round-robin}.

\subsection{Related Work}

\paragraph{Asymptotic fair division.}
Our work is closely related to the growing literature on asymptotic fair division \cite{Amanatidis2017,BaiGolz2022,BaiEC2022,Dickerson2014,Kurokawa2016,ManurangsiSuksompong2017,ManurangsiSuksompong2020,ManurangsiSuksompong2021,Suksompong2016}. 
Dickerson et al.~\cite{Dickerson2014} initiated the study of asymptotic fair division. They assumed that valuations are additive and drawn from a distribution with positive variances. 
Although the non-existence of an envy-free allocation also holds in this setting, Dickerson et al.~\cite{Dickerson2014} demonstrated that a welfare-maximizing algorithm for additive agents produces an envy-free allocation with a probability that approaches $1$ as $m$ goes to infinity when $m=\mathrm{\Omega}(n\log n)$. 
%

Following \cite{Dickerson2014}, Manurangsi and Suksompong~\cite{ManurangsiSuksompong2020} proved that, assuming that utilities are drawn from a polynomial-bounded distribution when $m$ is divisible by $n$, an envy-free allocation exists for agents with additive valuations with a probability that approaches $1$ as $m\to\infty$.
Moreover, Manurangsi and Suksompong~\cite{ManurangsiSuksompong2021} showed that under the assumption that utilities are drawn from a PDF-bounded distribution, agents have additive valuations and $m=\mathrm{\Omega}(n\log n/\log \log n)$, the round-robin algorithm returns an envy-free allocation with a probability that approaches $1$ as $n\to\infty$. Apart from requiring fewer items for establishing asymptotic envy-freeness, the round-robin algorithm has another advantage over the welfare-maximizing algorithm; it achieves an approximation of envy-freeness, specifically envy-freeness up to one item (EF1), for additive agents~\cite{Caragiannis2019}.

Bai and G\"{o}lz~\cite{BaiGolz2022} extend these results to the case where agents have asymmetric distributions when distributions are PDF-bounded.
Furthermore, Benad\`{e} et al.~\cite{BenadeHalpernPsomasVerma2023} demonstrated that the round-robin algorithm produces an SD envy-free allocation with a probability that approaches $1$ as $m\to\infty$ when agents have order-consistent valuation functions, items are renamed by a uniformly random permutation, $m$ is divisible by $n$, and $m=\omega(n^2)$.
Several papers have studied the asymptotic existence of allocations that satisfy other fairness notions, such as proportionality~\cite{Suksompong2016} or a maximin share guarantee~\cite{Amanatidis2017,Kurokawa2016}.

\paragraph{Asymptotic house allocation.}
In the one-to-one house allocation problem, Gan et al.~\cite{GanSuksompongVoudouris2019} studied the asymptotic existence of an envy-free allocation, which requires that there is no envy between every pair of individual agents. 
They demonstrated that when the agents' preferences are drawn uniformly at random and $m=\mathrm{\Omega}(n\log n)$, the probability that an envy-free allocation exists converges to $1$ as $n$ goes to infinity.
Furthermore, Manurangsi and Suksompong~\cite{ManurangsiSuksompong2021} showed that, if $m/n\geq\mathrm{e}+\varepsilon$, where $\mathrm{e}$ is Napier's constant and $\varepsilon>0$ is any constant, an envy-free allocation can be found with a probability that converges to $1$ as $n\to \infty$. Conversely, if $m/n\leq\mathrm{e}-\varepsilon$, there is no envy-free allocation with a probability that converges to $1$ as $n\to \infty$. 
In the house assignment problem, the cardinal model is equivalent to the ordinal setting because agents must compare only individual items. However, in our setting, classes must compare bundles of items. 
As mentioned in Section~\ref{sec:mainresults}, the asymptotic existence of a class envy-free matching where each agent is matched to exactly one item readily follows from the result in the house allocation problem~\cite{GanSuksompongVoudouris2019,ManurangsiSuksompong2021}, though such a matching may be wasteful.

\paragraph{Assignment valuations.}
Our problem can be formulated as a fair division problem with \emph{assignment valuations}, which may violate the additive assumption prevalent in fair division literature. Specifically, each class is represented by a ``meta'' agent, and the meta agent's valuation for each set of items is determined by the maximum total weight of an optimal matching between the class members and the items in a given bipartite graph.  
An assignment valuation is a generalization of additive valuations. That is, given an arbitrary additive valuation $v$ for $m$ items, an equivalent assignment valuation can be created by representing each agent $i$ with a class $i$ with $m$ copies of agents, with every copy having the same edge weight $v_i(j)$ towards each item $j$. 
However, the same construction cannot be used to represent the distributional model for additive valuations using assignment valuations.
%

\paragraph{Random assignment.}
We briefly review the literature on the theory of random assignments. 
Let $C_{n,m,r}$ be the minimum total weight of the matching with $r$ edges in a bipartite graph with $n$ and $m$ vertices on each side when edge weights are independently assigned from the exponential distribution with rate $1$. Here, we assume that $n\leq m$. 
When $n=m=r$, Walkup~\cite{Walkup1979} showed that the expected value is bounded above when $n$ goes to infinity. 
Following numerous papers on experimental results and improved bounds~(see the introduction in~\cite{Wastlund2009} for details), Karp~\cite{Karp1987} improved the upper bound and showed that the expected value is smaller than $2$ for any $n$. Aldous~\cite{Aldous1992,Aldous2001} showed that $\mathbb{E}[C_{n,n,n}]$ converges to $\frac{\pi^2}{6}$ as $n$ goes to infinity. 
For a more general combination of $n,m$, and $r$, Linusson and W\"{a}stlund~\cite{LinussonWastlund2001} and Nair, Prabhakar and Sharma~\cite{NairPrabhakarSharma2006}
obtained a concrete formula for the expected minimum total weight of a matching
given by $\mathbb{E}[C_{n,m,r}]  = \sum_{i=1}^r \frac{1}{n} \sum_{j=0}^{i-1} \frac{1}{m-j}$.
W\"{a}stlund~\cite{Wastlund2009} provided a concise and elegant proof for this result, by analyzing the expected difference between the minimum weight of matching with $r$ edges and that with $r-1$ edges and showing that $\mathbb{E}[C_{n,m,r}] - \mathbb{E}[C_{n,m,r-1}] = \frac{1}{n}\sum_{j=0}^{r-1}\frac{1}{m-j}$.
Frieze and Johansson~\cite{FriezeJohansson2017} and Frieze~\cite{Frieze2021} explored a similar approach for random bipartite graphs and non-bipartite graphs, and W\"{a}stlund~\cite{Wastlund2012} and Larsson~\cite{Larsson2021} extended the result in~\cite{Wastlund2009} to more general distributions in some pseudo-dimension.
Recently, the random assignment problem where edge weights are drawn independently from a standard Gaussian distribution is investigated~\cite{MordantSegers2021}.


\section{Model}\label{sec:Preliminaries}

We use $[k]$ to denote the set $\{1, 2,\ldots, k\}$. Let $N=[n]$ be the set of $n$ agents, and $I=[m]$ be the set of $m$ items. 
The set of agents $N$ is partitioned into $k$ classes, labeled as $N_1, N_2, \ldots, N_k$. 
Let $n_p = |N_p|$ for each class $p$. We assume that $n_p \geq 1$ for every $p\in [k]$. Each $N_p$ is referred to as class $p$. 
We call a subset of $I$ a \textit{bundle}. 

We consider a matching problem where each item in $I$ is matched to at most one agent in $N$, and each agent receives at most one item. 
Each agent $i \in N$ is endowed with a non-negative utility $u_i(j)$ for every item $j \in I$ where $u_i(j)$ ranges within the interval $[0,1]$. We assume that $u_i(j)$ is drawn from a distribution over $[0,1]$. Detailed assumptions on distributions are presented in Section~\ref{sec:Preliminaries}. 

We define a complete bipartite graph $G=(N\cup I,E)$,  where the set of agents in $N$ forms the left vertices and the set of items in $I$ forms the right vertices. Here, $E$ denotes the set of edges. 
We consider the weights of edges where the weight of edge $\{i,j\}\in E$ is given by $u_i(j)$ for each $i\in N$ and $j\in I$. 
A \textit{matching} $M$ in $G$ is defined as a set of edges wherein each vertex appears at most one edge of $M$. 
For $S \subseteq N$, let $M(S)$ be the set of items which are assigned to some agent in $S$ by matching $M$, i.e., $M(S)=\{j \in I \mid \exists\, i \in S: \{i,j\} \in M\}$. Each matching $M$ induces an allocation that assigns the bundle $M(N_p)$ to every class $p$. 
For bundle $I'\subseteq I$, let $\mathcal{M}(N_p,I')$ denote the set of possible matchings between $N_p$ and $I'$ in $G$.
We define the total utility obtained by class $p$ under matching $M$ as $u_p(M) = \sum_{\{i,j\} \in M, i\in N_p}u_{i}(j)$. 

To define envy between classes, we formalize how much hypothetical value each class can derive from a bundle allocated to another class. To do this, we introduce an \emph{assignment valuation}.
\begin{definition}[Assignment valuation]\label{def:assignment_val}
    An \textit{assignment valuation} $v_p(I')$ of class $p$ for a bundle $I' \subseteq I$ is defined as the maximum total weight of a matching between the agents in $N_p$ and the items in $I'$. Namely, $v_p(I')$ is given by $\max_{M\in \mathcal{M}(N_p,I')}\sum_{\{i,j\}\in M}u_i(j)$.
\end{definition}
It is worth noting that the assignment valuation $v_p(I')$ for bundle $I'\subseteq I$ is upper bounded by the size $n_p$ of each class $p\in[k]$ since $u_i(j)\le 1$ for all $i\in N$ and $j\in I$.
Here, each $v_p(I')$ can be computed in polynomial time by computing a maximum-weight matching in the given bipartite graph with edge weights.
See Section~9 in~\cite{LovaszPlummer2009}. 

Next, we introduce the concept of fairness among classes---\emph{class envy-freeness}. This principle requires that the total utility each class receives must be greater than or equal to the maximum total utility that the class can derive from the items allocated to other classes. 
\begin{definition}[Class envy-freeness]\label{def:ef_for_assignment_val}
For a matching $M$, we say that class $p$ \emph{envies} class $q$ if $u_p(M) < v_p(M(N_q))$. 
    A matching $M$ is called \textit{class envy-free} if no class envies another class, i.e., $u_p(M) \geq v_p(M(N_q))$ holds for every pair $p,q \in [k]$ of two distinct classes.
\end{definition}

If we allow each class to optimally reassign items within the members of the class, then the class would select a maximum-weight matching between the members of the class and their bundle.
In such a scenario, the class envy-freeness requirement is equivalent to the above-mentioned definition, where the left-hand side $u_p(M)$ is replaced by $v_p(M(N_p))$.

As is observed in~\cite{Benabbou2019,HosseiniHuangIgarashiShah2022}, unfortunately, there exists an input where a class envy-free matching may not exist without allowing us to dispose items. Thus, the following approximation of class envy-freeness has been considered in \cite{HosseiniHuangIgarashiShah2022}. 
A matching $M$ is \emph{$\alpha$-class envy-free matching up to one item (CEF1)} if for every pair of classes $p, q \in [k]$, either class $p$ does not envy class $q$, or there exists an item $j\in M(N_q)$ such that $\alpha^{-1} \cdot u_p(M) \geq v_p(M(N_q)\setminus \{j\})$. If $\alpha=1$, we call such a matching \emph{CEF1}~\cite{HosseiniHuangIgarashiShah2022}.

Next, we define a measure of efficiency called \emph{non-wastefulness}~\cite{Benabbou2019}. Non-wastefulness requires that valuable items are not wasted. 
\begin{definition}[Non-wastefulness]\label{def:nowasteful}
    An item $j \in I$ is said to be \textit{wasted} for a matching $M$ if either
    \begin{enumerate}[label=$(\mathrm{\alph*})$]
    \item\label{item:consition-a-of-non-wasteful} item $j$ is an unallocated and can increase the total utility of some class, i.e., we have $j \notin M(N)$ and $v_p(M(N_p) \cup j) - v_p(M(N_p)) >0$ for some class $p$, or 
    \item\label{item:consition-b-of-non-wasteful} item $j$ can be reallocated from class $q$ to class $p$ in a way that increases the total utility of class $p$ without reducing the total utility of class $q$, i.e., there exists class $q$ such that $j \in M(N_q)$ and $v_p(M(N_q)) - v_p(M(N_q)\setminus \{j\})=0$
    but $v_p(M(N_p) \cup j) - v_p(M(N_p))>0$ for some class $p$.
    \end{enumerate}
    A matching is \textit{non-wasteful} if no item is wasted.
\end{definition}
As mentioned in Introduction, a class envy-free matching that satisfies non-wastefulness may not exist. For example, consider the case of a single item being divided between two classes. 
Additionally, we remark that if we do not impose non-wastefulness, a CEF1 matching that allocates all items always exists and can be found in polynomial time using the envy-graph algorithm introduced by Lipton~\cite{Lipton2004}. 
However, the matching produced by the envy-graph algorithm may not satisfy non-wastefulness as pointed out by Benabbou et al.~\cite{Benabbou2019}.


\paragraph{Distributions.}
For each agent $i \in N$ and item $j \in I$, 
the utility $u_i(j)$ is independently drawn from a given distribution $\mathcal{D}$ supported on $[0,1]$. 
Let $f_\mathcal{D}$ and $F_\mathcal{D}$ denote the probability density function (PDF) and the cumulative distribution function of $\mathcal{D}$, respectively.
A distribution is said to be \emph{non-atomic} if it does not assign a positive probability to any single point. 

We say that a distribution $\mathcal{D}$ is \emph{$(\alpha,\beta)$-PDF-bounded} for constants $0 < \alpha \leq \beta$ if it is non-atomic and $\alpha \leq f_{\mathcal{D}}(x) \leq \beta$ for all $x \in [0,1]$. 
When $\alpha=\beta=1$, $\mathcal{D}$ represents the uniform distribution over $[0,1]$ since $f_{\mathcal{D}}(x)=1$ for all $x \in [0,1]$. 
The PDF-boundedness assumption is introduced by~\cite{ManurangsiSuksompong2021} in the context of fair division as a natural class of distributions, which includes, for example, the uniform distribution and the truncated normal distribution.

Let $\text{Exp}(\lambda)$ denote the exponential distribution with rate $\lambda$ over $[0, \infty)$. 
Furthermore, let $\mathrm{ReExp}(\lambda)$ denote a distribution with the density function $f_{\lambda}(x) = \lambda \mathrm{e}^{-\lambda (1-x)}$ on the interval $(- \infty,1]$. We call this probability distribution the \emph{reversed exponential distribution}, which mirrors the exponential distribution $\text{Exp}(\lambda)$ across the line $x=1/2$. The cumulative distribution of $\mathrm{ReExp}(\lambda)$ is given by $F_{\mathrm{ReExp}(\lambda)}(x)=\mathrm{e}^{-\lambda(1-x)}$.
We say that an event occurs \emph{almost surely} if it occurs with probability $1$.


\paragraph{Known results on maximum-weight matchings.}
We present several known results on maximum-weight matchings in a bipartite graph with random edge weights that are used throughout the entire paper. 
For completeness, we show the proofs in Appendix~\ref{app:nesting_lemma}. 

Let $H$ be a complete bipartite graph with bipartition $(A,B)$. 
For $A'\subseteq A$ and $B'\subseteq B$, let $H[A',B']$ denote the subgraph of $H$ induced by $A'$ and $B'$. 
We first present the following lemma, which can be proven by a proof similar to that of the isolation lemma~\cite{Jukna2001,Spencer1996}. Note that after edge weights on $H$ have been sampled, we can select $A'$ and $B'$ while referring to the values of those edge weights.
\begin{restatable}
{lemma}{UniquenessMaximumMatching}\label{lemma:uniqueness_of_maximum_weight_matching}
    Let $H$ be a complete bipartite graph with bipartition $(A,B)$ whose edge weights are drawn independently from a non-atomic distribution on $[0,1]$. Let $A'\subseteq A$ and $B'\subseteq B$. 
    Then, no pair of distinct matchings has the same total weight in $H[A',B']$ almost surely.
\end{restatable}
Lemma~\ref{lemma:uniqueness_of_maximum_weight_matching} implies that, almost surely, a maximum-weight matching of a fixed size in $H[A',B']$ is uniquely determined. 
We next explain the nesting lemma, which follows from Lemma $3$ in \cite{BuckChanRobbins2002} or Lemma $2.1$ in \cite{Wastlund2009}. 
\begin{restatable}[The nesting lemma]{lemma}{NestingLemma}\label{lemma:nesting_lemma}
 Let $H$ be a complete bipartite graph with bipartition $(A,B)$ whose edge weights are drawn independently from a non-atomic distribution on $[0,1]$. Let $A'\subseteq A$ and $B'\subseteq B$. 
 Then, every vertex that appears as an element of the maximum-weight matching with $r-1$ edges in $H[A',B']$ also appears as an element of the maximum-weight matching with $r$ edges in $H[A',B']$ almost surely. Here, $1<r\leq \min(|A'|,|B'|)$.
\end{restatable}

We introduce notations and definitions related to matchings in a bipartite graph. For a bipartite graph $H$ and each vertex $i$ that appears in a matching $M$ of $H$ (namely $\{i,j\}\in M$ for some $j$), we denote by $M(i)$ the vertex matched to $i$ under $M$. 
An \emph{alternating} path $P$ (resp. a cycle $C$) of matching $M$ in bipartite graph $H$ is a path (resp. a cycle) in $H$ where, for every pair of consecutive edges on $P$, one of them is in $M$ and the other one is not in $M$.


\section{Maximum-Weight Matching}\label{sec:mainresults}

We introduce our first result, which states if the number of items is quadratically large, then the probability that a maximum-weight matching is class envy-free approaches $1$ as $m\to \infty$. 

\begin{restatable}{theorem}{ThmMaxmumWeightMatching}\label{thm:maxmatching}
    Suppose that $\mathcal{D}$ is non-atomic and there exists a constant $c>0$ with $m/ \log m \ge c \cdot (k^2 \max_{p \in [k]} n_p)^2 $. 
    Then, as $m \to \infty$, the probability that a maximum-weight matching is class envy-free approaches $1$.
\end{restatable}
To prove Theorem~\ref{thm:maxmatching}, we rely on the following simple observation: if each class desires a disjoint bundle, then the maximum-weight matching is class envy-free. 
\begin{restatable}{lemma}{LemDistinctBundle}\label{lem:bundle_distinct}
Suppose that there exist $k$ disjoint bundles $I_1,I_2,\ldots,I_k$ such that such that for every $p\in [k]$, $I_p$ is a most favorite bundle of class $p$, i.e., $I_p \in \mathrm{argmax}_{I'\subseteq I, |I'|=n_p} v_p(I')$. Then, any maximum-weight matching in $G$ is class envy-free and non-wasteful.
\end{restatable}
\begin{proof} 
    Let $M$ denote the matching that results from optimally assigning each $I_p$ to class $N_p$ for $p =1,2,\ldots,k$. This means that $u_p(M) = v_p(M(N_p))$ for $p =1,2,\ldots,k$.

    Take any maximum-weight matching $M_{\mathrm{OPT}}$ in $G$.
    Let $w(M)$ and $w(M_{\mathrm{OPT}})$ be the total weights of $M$ and $M_{\mathrm{OPT}}$ each. 
    Then since $u_p(M_{\mathrm{OPT}}) \leq u_p(M)$ for every $p \in [k]$,
    \[
        w(M_{\mathrm{OPT}})= \sum_{p\in [k]}u_p(M_{\mathrm{OPT}})  \leq \sum_{p\in [k]} u_p(M) = w(M).
    \]
    Thus, $M$ is a maximum-weight matching, and we have $ w(M_{\mathrm{OPT}}) = w(M)$. Therefore, we obtain $u_p(M_{\mathrm{OPT}}) = u_p(M)$ for every $p\in [k]$.
    Hence, each class is matched to its most favorite bundle under $M_{\mathrm{OPT}}$, and therefore matching $M_{\mathrm{OPT}}$ is class envy-free. Clearly, $M_{\mathrm{OPT}}$ is non-wasteful.
\end{proof}

\begin{proof}[Proof of Theorem~\ref{thm:maxmatching}]
%
By Lemma~\ref{lemma:uniqueness_of_maximum_weight_matching}, a maximum-weight matching in $G$ of size $n$ is uniquely determined almost surely.
Since each edge weight is drawn from the same non-atomic distribution, we have that $\mathbf{Pr}\left[ \text{argmax}_{I''\subseteq I, |I''|=n_p} v_p(I'') = I' \right] = \frac{1}{{m \choose n_p}}$
for all $p\in [k]$ and $I'\subseteq I$ with $|I'|=n_p$.

Let $\mathcal{A}$ be the event that there exists a class envy-free and non-wasteful matching.
Let $\mathcal{B}$ be the event that there exist $k$ disjoint bundles $I_1,I_2,\ldots,I_k$ such that $I_p \in \mathrm{argmax}_{I'\subseteq I, |I'|=n_p} v_p(I')$ for every $p\in [k]$. 
By Lemma~\ref{lem:bundle_distinct}, if any two classes' the most favorite bundles are disjoint, then there exists a class envy-free and non-wasteful matching. Then, we have $\mathbf{Pr}[\mathcal{A}] \geq \mathbf{Pr}[\mathcal{B}]$. 
Let $P$ denote the set of partitions of the $m$ items into disjoint bundles of sizes $n_1,n_2,\ldots,n_k$.
We provide a lower bound for $\mathbf{Pr}[\mathcal{B}]$ as follows.
\begin{align*}
    \mathbf{Pr}[\mathcal{B}] 
    &= \mathbf{Pr}\left[ \text{there are $k$ disjoint bundles $I_1,I_2,\ldots,I_k$ s.t. $\underset{I'\subseteq I,\, |I'|=n_p} {\operatorname{argmax}}  v_p(I') = I_p$ for every $p\in [k]$} \right]  \\
    &= 
        \sum_{(I_1,I_2,\ldots,I_k)\in P} \mathbf{Pr}\left[ \text{$\underset{I'\subseteq I,\, |I'|=n_p} {\operatorname{argmax}} v_p(I') = I_p$ for every $p\in [k]$}\right]  \\
    &= 
        \sum_{(I_1,I_2,\ldots,I_k)\in P} \frac{1}{{m \choose n_1}}\cdot \frac{1}{{m \choose n_2}}\cdot \cdots\cdot  \frac{1}{{m \choose n_k}}  \\
    &= 
        \frac{{m \choose n_1}}{{m \choose n_1}}\cdot \frac{{m-n_1 \choose n_2}}{{m \choose n_2}}\cdot \cdots\cdot  \frac{{m-\sum_{i=1}^{k-1}n_i \choose n_k}}{{m \choose n_k}} \\
    &\ge  
        \left( 1- \frac{n_2}{m-n_1+1} \right)^{n_1} \cdot  \left( 1- \frac{n_3}{m-n_1-n_2+1} \right)^{n_1+n_2} \cdot
    \cdots 
        \cdot  \left( 1- \frac{n_k}{m-{\sum_{i=1}^{k-1}n_i}+1} \right)^{\sum_{p=1}^{k-1}n_p} \\
    &\ge 
        \exp \left( - \frac{n_1 n_2}{m-n_1+1} - \frac{(n_1+n_2)n_3}{m-n_1-n_2+1} - \cdots -\frac{(\sum_{i=1}^{k-1}n_i)n_k}{m-\sum_{p=1}^{k-1}n_p+1} \right) \\
    &\ge  
        \exp \left( - \frac{k^2 (\max_{p\in [k]} n_p)^2}{m-\sum_{p=1}^{k-1}n_p+1} \right)\\
    &\ge  
        \exp \left( - \frac{1/c \cdot m\log m}{m-\sqrt{1/c \cdot m\log m}+1} \right).
\end{align*}
For the last inequality, we use $k^2 (\max_{p\in [k]} n_p)^2 \le 1/c \cdot m\log m$ and $\sum_{p=1}^{k-1}n_p \le \sqrt{1/c \cdot m\log m}$.
From this, we have $\mathbf{Pr}[\mathcal{A}] \to 1$ as $m \to \infty$.
\end{proof}

We note that the asymptotic existence of a class envy-free matching where every agent obtains exactly one item can be readily derived from the existing result on the one-to-one house allocation problem given by~\cite{ManurangsiSuksompong2021}. They showed that an envy-free assignment can be obtained by considering a greedy algorithm. 
This algorithm selects, in each step, an agent who has not yet been assigned an item. This agent then chooses their favorite item from those that have not been discarded (including items that have already been assigned to other agents).
If the selected item was previously chosen by another agent in an earlier step, it is removed from further consideration.  
When there are sufficiently many items, this algorithm asymptotically produces a matching where each agent receives their most preferred item among those that were not discarded, resulting in a class envy-free matching. Formally, see Proposition~\ref{complete_class_envy-free} in  Appendix~\ref{app:for_contributions}.


As previously mentioned, both the maximum-weight matching and the matching produced by the greedy algorithm of \cite{ManurangsiSuksompong2021} can be inherently unfair for worst-case inputs. This raises the question of whether there exists a matching mechanism that is fair for both worst-case and average-case inputs.


\section{Round-Robin Algorithm}\label{sec:round-robin}

We next present our second result that a round-robin algorithm, tailored for assignment valuations and presented as Algorithm~\ref{alg:main_alg_1}, produces a $1/2$-CEF1 non-wasteful matching that is class envy-free asymptotically. 
By the recent results of \cite{Amanatidis2023,Montanari2024}, we can show that the round-robin algorithm always returns a matching satisfying an approximate notion of class envy-freeness, $1/2$-CEF1. Furthermore, we prove that it is non-wasteful. 

\begin{restatable}{proposition}{RoundRobinNWEF}\label{proposition:1_2_ef1_non_waste}
    The matching produced by the round-robin algorithm is $1/2$-CEF1
    and non-wasteful. 
\end{restatable}
Moreover, we demonstrate that the matching produced by the round-robin algorithm is asymptotically class envy-free.

\begin{restatable}{theorem}{MainTheorem}\label{thm:main_theorem_1}
    Suppose that $\mathcal{D}$ is $(\alpha,\beta)$-PDF-bounded, and the following three conditions~\ref{item:assumption-b},~\ref{item:assumption-c} and~\ref{item:assumption-a} hold.
    \begin{enumerate}[label=$(\mathrm{\alph*})$]
        \item\label{item:assumption-b} The number of items $m$ is sufficiently large such that $m \geq k\cdot \max_{p\in [k]} (n_p+2)$,
        \item\label{item:assumption-c} the class sizes are almost proportional to the total population; more precisely, there exists a constant $C>0$ such that $n\leq C\cdot (\min_{p\in [k]} n_p)^{5/4}$, and
        \item\label{item:assumption-a} the number of classes $k$ satisfies that $k> \max\left(\frac{1}{2\alpha},\frac{\beta}{\alpha^2}\right)$ and $k^2=O(n^{1/6})$.
    \end{enumerate}
    Then, as $n\to \infty$, the probability that the round-robin algorithm produces a class envy-free and non-wasteful matching converges to $1$.
\end{restatable}
Note that the first condition~\ref{item:assumption-b} is slightly stronger than the condition where $m\ge n$.
The third condition~\ref{item:assumption-a} is very mild for some distributions, e.g. for uniform distributions, this condition is equivalent to $k > 1$ since $\alpha=\beta=1$. Moreover, we consider the situation where the number of groups are not large comparing the number of agents in each class.
A perhaps more intuitive but stronger condition of~\ref{item:assumption-c} is the case where the total number of agents is within a constant factor of the minimum size of a class, i.e., $n =  c \cdot \min_{p\in [k]} n_p$, where $c \geq 1$ is a constant. This is relevant in scenarios where the class sizes under consideration are proportional to the total population, such as gender groups, ethnic groups, and groups of people with the same political interests.

\paragraph{Algorithmic description.} 
Let us now explain the round-robin algorithm (Algorithm~\ref{alg:main_alg_1}). 

\begin{algorithm}[H]
    \caption{The round-robin algorithm for classes with assignment valuations}
    \label{alg:main_alg_1}
    \begin{algorithmic}[1]
        \REQUIRE $N=N_1\cup N_2\cup\cdots\cup N_k$, $I$, $\{u_i(j)\}_{i\in N, j\in I}$
        \ENSURE Matching $M$
        \STATE $M \leftarrow \emptyset$, $I_0 \leftarrow I$, $r \leftarrow 1$, and $M^0_{p} \leftarrow \emptyset$ $\forall p \in [k]$
        \WHILE{there is a remaining item which some class desires \label{line:while_start}}
        \FOR{$p=1,2,\ldots,k$}
        \IF{there is a remaining item which class $p$ desires}
        \STATE Let $j^{r}_p \in \textrm{argmax}_{j \in I_0} v_p(M(N_p) \cup \{j\})- v_p(M(N_p))$\label{line:higestmarginal}
        \STATE $M^r_p \leftarrow$ the maximum-weight matching between $N_p$ and $M(N_p)\cup\{j^{r}_p\}$\label{line:updatematching}
        \STATE $M \leftarrow (M\setminus M_p^{r-1}) \cup M^r_p$\label{line:update_whole_matching}
        \STATE $I_0 \leftarrow I_0 \setminus \{j^{r}_p\}$
        \ENDIF
        \ENDFOR
        \STATE $r \leftarrow r+1$
        \ENDWHILE \label{line:while_end}
        \RETURN $M$.
    \end{algorithmic}
\end{algorithm}

Each iteration of the {\bf while} loop (Lines~\ref{line:while_start}--\ref{line:while_end}) in Algorithm~\ref{alg:main_alg_1} is referred to as a \textit{round}. In each round $r$, each class selects its most preferred item, which has the highest marginal utility to the current bundle (Line~\ref{line:higestmarginal}) 
and updates its matching to create a new maximum-weight matching with an additional edge (Line~\ref{line:updatematching}). 
If a class encounters several items with the highest marginal utility, it selects one item arbitrarily among them.
In Line~\ref{line:update_whole_matching}, the matching between $N$ and $I$ is updated to reflect the new item acquired by the class $p$. Observe that in Algorithm~\ref{alg:main_alg_1}, we allow each class to optimally reassign items within its members in each round, thereby selecting a maximum-weight matching between its members and the items allocated to them. 
Consequently, the total utility $u_p(M)$ that class $p$ receives under Algorithm~\ref{alg:main_alg_1} is $v_p(M(N_p))$.

In contrast to the round-robin algorithm for additive valuations, 
Algorithm~\ref{alg:main_alg_1} may not produce CEF1 matching; in fact, the factor of $1/2$ is the best that can be achieved by the round-robin algorithm; see Example 1 of \cite{HosseiniHuangIgarashiShah2022}.

\begin{example}\label{example:round-robin}
Consider an instance with two classes, each consisting of two agents, and four items in Table~\ref{table_house_allocation}.
The unique maximum-weight matching for this instance is given by $\{\{i_1,j_1\}$, $\{i_2,j_4\}$, $\{i_3,j_2\}$, $\{i_4,j_3\}\}$. However, this matching does not satisfy class envy-freeness as the second class receives a total utility of $3$ despite having a maximum-weight matching of value $5$ with the bundle allocated to the first class.
In contrast, Algorithm~\ref{alg:main_alg_1} produces matching $\{\{i_1,j_1\},\{i_2,j_2\},$ $\{i_3,j_4\},\{i_4,j_3\}\}$, which can be easily checked to be both class envy-free and non-wasteful.
\begin{table}[h]
\begin{center} 
    \caption{Maximum-weight matching (left) and the matching produced by round-robin algorithm (right).}
    \label{table_house_allocation}
    \begin{tabular}{cc}
    \begin{tabular}{|c|c||c|c|c|c|}
      \hline
      & & $j_1$ & $j_2$ & $j_3$ & $j_4$  \\
      \hline
      \multirow{2}{*}{$N_1$}  &$i_1$ & \circled{5} & $0$ & $0$ & $0$  \\ 
      \cline{2-6}
       & $i_2$ & $0$ & $1$ & $0$ & \circled{5}  \\
      \hline
      \multirow{2}{*}{$N_2$}  & $i_3$ & $2$ & \circled{1} & $0$ & $3$ \\
      \cline{2-6}
       & $i_4$ & $1$ & $0$ & \circled{2} & $0$  \\
       \hline
    \end{tabular}
    \begin{tabular}{|c|c||c|c|c|c|}
      \hline
      & & $j_1$ & $j_2$ & $j_3$ & $j_4$  \\
      \hline
      \multirow{2}{*}{$N_1$}  &$i_1$ & \circled{5} & $0$ & $0$ & $0$  \\ 
      \cline{2-6}
       & $i_2$ & $0$ & \circled{1} & $0$ & $5$  \\
      \hline
      \multirow{2}{*}{$N_2$}  & $i_3$ & $2$ & $1$ & $0$ & \circled{3} \\
      \cline{2-6}
       & $i_4$ & $1$ & $0$ & \circled{2} & $0$  \\
       \hline
    \end{tabular}
    \end{tabular}
    \end{center}
\end{table}
\end{example}


\subsection{Outline of the Proof}
The remainder of this paper is devoted to proving Theorem~\ref{thm:main_theorem_1}.
Throughout this section, let $M$ denote the matching produced by Algorithm~\ref{alg:main_alg_1} and $M^{r}$ denote the matching at the end of round $r$ in the algorithm. 
Furthermore, $M^{r}(N_p)$ is defined as the set of items matched to an agent in class $p$ under matching $M^r$.
Let $I_p^{r}$ denote the set of remaining items just before class $p$ selects an item in round $r$, i.e., $I_p^{r} = I\setminus \bigl(M^{r}(N_1)\cup \cdots \cup M^{r}(N_{p-1})\cup M^{r-1}(N_p)\cup \cdots \cup M^{r-1}(N_{k}) \bigr)$. 
We fix two classes $p$ and $q$ and analyze the behavior of the following random variables:
\[
    X_p = v_p(M(N_p))\quad \text{and}\quad X_{pq} = v_p(M(N_q)).
\]

To prove Theorem~\ref{thm:main_theorem_1}, we first examine expected marginal weights of maximum-weight matchings (Lemma~\ref{lemma:difference_between_two_expected_values}) in Section~\ref{sec:roundrobin:prelim}.
This lemma claims that the difference in expected values can be expressed in terms of the probability that the special vertex belongs to a maximum matching in a modified graph. 
In the proof of Lemma~\ref{lemma:difference_between_two_expected_values}, we adopt a technique pioneered by \cite{Wastlund2005,Wastlund2009}. 
These works introduced an additional vertex and connected it to every vertex on the other side by an edge with weights following an exponential distribution to analyze the expected marginal weight of minimum weight matchings; we adapt this technique to the context of maximum-weight matchings. 

By utilizing Lemma~\ref{lemma:difference_between_two_expected_values}, we establish bounds on the expected values of $X_p$ and $X_{pq}$ (Lemmas~\ref{claim:expected_diff_p} and~\ref{claim:expected_diff_q}).
In the proof, we analyze the probabilities that special vertices belong to maximum matchings in graphs between $N_p$ and $M(N_p)$ or $M(N_q)$.
Due to the uniqueness of the maximum-weight matching (Lemma~\ref{lemma:uniqueness_of_maximum_weight_matching}), we consider the unique augmenting path updating it in each round.

In Section~\ref{sec:final_together_proof},
by Lemmas~\ref{claim:expected_diff_p} and~\ref{claim:expected_diff_q}, and by demonstrating that the edge weights in maximum-weight matchings between $N_p$ and $M(N_q)$ are sufficiently ``heavy'' ((b) in Lemma~\ref{lem:no_heavy_edge}), 
we prove that the difference in the expected values of $X_p$ and $X_{pq}$ is lower-bounded (Lemma~\ref{lemma:expected_diff}). 
Finally, we achieve stochastic concentrations on the expectations of $X_p$ and $X_{pq}$, establishing Theorem~\ref{thm:main_theorem_1}.


\subsection{Expected Marginal Weight of a Maximum-Weight Matching}\label{sec:roundrobin:prelim}
In this section, we show Lemma~\ref{lemma:difference_between_two_expected_values} by examining the difference in expected values arising from a maximum-weight matching of consecutive sizes.
Consider a complete bipartite graph $H$ with the left set $A$ of vertices and the right set $B$ of vertices. The weights of all edges in $H$ are derived from non-atomic distributions over the interval $[0,1]$, without the requirement for these weights to be drawn independently. We make the following assumptions: 
\begin{assumption}\label{assumption:no_distinct}
    No two distinct matchings have the same total weight almost surely in $H$.
\end{assumption}
\begin{assumption}\label{assumption:belongs_maximum_weight_matching}
For every size $r =1,2,\ldots,\min\{|A|,|B|\}$ and for any pair of vertices $i,i' \in A$, the probability that $i$ belongs to the maximum-weight matching of size $r$ in $H$ is the same as that for $i'$.
\end{assumption}

Now, we modify $H$ by introducing a new vertex $\hat{j}$ to $B$, and create edges $\{i,\hat{j}\}$ for all $i\in A$ (see Figure~\ref{fig:bipartite_graph}). 
The weight of each edge 
$\{i,\hat{j}\}$ for $i\in A$ is independently drawn from $\mathrm{ReExp}(\lambda)$ on $(-\infty, 1]$ where $0<\lambda\leq 1$. Let $\hat{H}$ denote the modified bipartite graph. 
%
Later, we will consider the limit probability of $\hat{j}$ being included in a maximum-weight matching when $\lambda$ converges to $0$.
Let $\hat{B} = B \cup \{\hat{j}\}$. 
Note that edges with negative weights between $A$ and $\{\hat{j}\}$ will not be part of any maximum-weight matching of size up to $\min\{|A|,|B|\}$ in $\hat{H}$ since the edges in $H$ are non-negative. 
Let $\hat{B}^{r}$ represent the set of vertices in $\hat{B}$ under the maximum-weight matching of size $r$ between $A$ and $\hat{B}$. Note that no two distinct matchings have the same total weight almost surely in $\hat{H}$ by non-atomicity of $\mathrm{ReExp}(\lambda)$, by the assumption that the weights of edges incident to ${\hat j}$ are drawn independently, and by the assumption that no two distinct matchings in $H$ have the same total weight.

Let $X^r$ denote the maximum-weight of a matching of size $r$ in $H$.
Lemma~\ref{lemma:difference_between_two_expected_values} states that the difference in expected values of $X^r$ and $X^{r-1}$ can be expressed in terms of the probability of $\hat{j}$ being included in $\hat{B}^r$.

\begin{figure}[tb]
    \centering
\begin{tikzpicture}[line width=1pt, scale=0.7]
\draw (0,0) rectangle (1,3) node[midway] {$A$};
\draw (3,0) rectangle (4,3) node[midway] {$B$};
\draw[decoration={brace,amplitude=7pt,raise=2pt},decorate] (4.2,3) -- node[right=8pt] {$\hat{B}$} (4.2,-0.75);

\node at (3.5,-0.5) {$\hat{j}$};
\draw[densely dotted] (0.7,2.5) -- (3.3,-0.4);
\draw[densely dotted] (0.7,1.5) -- (3.3,-0.4);
\draw[densely dotted] (0.7,0.5) -- (3.3,-0.4);

\draw[rounded corners=10pt, ultra thin] (-0.5,-1.0) rectangle (5.7,3.5);
\node at (-1.3,3) {$\hat{H}$};
\end{tikzpicture}
    \caption{The modified graph $\hat{H}$ with the additional vertex $\hat{j}$.}
    \label{fig:bipartite_graph}
\end{figure}
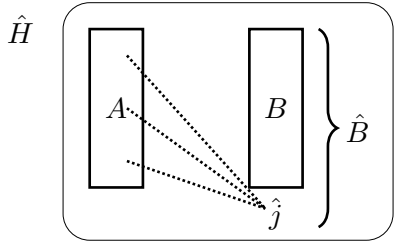

\begin{lemma}\label{lemma:difference_between_two_expected_values}
    For every $r=1,2,\ldots,\min(|A|,|B|)$, we have
     \begin{equation}\label{eq:expected difference_lemma_4}
        \mathbb{E}[X^r]-\mathbb{E}[X^{r-1}] = 1- \frac{1}{r}\lim_{\lambda \to 0}\frac{1}{\lambda}\mathbf{Pr}\left[\hat{j} \in \hat{B}^{r}\right].
    \end{equation}
\end{lemma}

\begin{proof}
    Let $W(i,\hat{j})$ denote the random variable representing the weight of the edge $\{i,\hat{j}\}$ for each $i \in A$.
Let $A^{r}$ denote the set of vertices in $A$ that are matched under the maximum-weight matching of size $r$ in $H$. By 
The maximum-weight matching of size $r$ between $A^r$ and $\hat{B}$ is denoted by ${\hat M}$, and for each $i \in A^r$, let $X^{r-1}_i$ be the maximum-weight of matchings of size $r-1$ between $A^r \setminus \{i\}$ and $B$.

Select $i$ from $A^{r}$ uniformly at random.
By definition, the maximum-weight of a matching of size $r$ between $A^r$ and $\hat{B}$ under the constraint that the edge $\{i,\hat{j}\}$ is included is $X^{r-1}_i+W(i,\hat{j})$. 
We claim that $\mathbf{Pr}\big[ X^{r-1}_i+W(i,\hat{j}) > X^r \big] = \mathbf{Pr}\big[\{i, \hat{j}\} \in {\hat M}\big] + O(\lambda^2).$
If edge $\{i,\hat{j}\}$ is included in ${\hat M}$, 
we have $X^{r-1}_i+W(i,\hat{j}) > X^r$ by Assumption~\ref{assumption:no_distinct}. Thus, $\mathbf{Pr}[\{i, \hat{j}\} \in {\hat M}] \leq \mathbf{Pr}[ X^{r-1}_i+W(i,\hat{j}) > X^r ]$. 

Next, suppose that $W(i,\hat{j}) > X^r - X^{r-1}_i$. Then, ${\hat M}$ must include an edge incident to $\hat{j}$. 
If no $i' \in A^r$ other than $i$ satisfies the inequality $W(i',\hat{j}) > X^r - X^{r-1}_{i'}$,  
then we get $\{i, \hat{j}\} \in {\hat M}$. 
{Let $\mathcal{F}_i$ denote the event that there exists no vertex $i' \neq i$ such that it satisfies $X^{r-1}_{i'} + W(i',\hat{j}) > X^r$. 
}
By the above argument, we have
\begin{align*}
    \mathbf{Pr}\big[ X^{r-1}_{i}+W(i,\hat{j}) > X^r \big] 
    &= 
        \mathbf{Pr}\big[ (X^{r-1}_{i}+W(i,\hat{j}) > X^r) \land \mathcal{F}_i\big]
    +
        \mathbf{Pr}\big[ (X^{r-1}_{i}+W(i,\hat{j}) > X^r) \land \overline{\mathcal{F}_i}\big] \\
    &\le 
        \mathbf{Pr}\big[\{i, \hat{j}\} \in {\hat M}\big]
    +
        \mathbf{Pr}\big[ (X^{r-1}_{i}+W(i,\hat{j}) > X^r) \land \overline{\mathcal{F}_i}\big] .
\end{align*}
The right term of the above inequality can be bounded as follows:
\begin{align*}
    \mathbf{Pr}\big[ X^{r-1}_{i}+W(i,\hat{j}) > X^r ~\land~ \overline{\mathcal{F}_i}\big] 
    &\leq 
        \sum_{i' \in A^{r}\setminus \{i\}}\mathbf{Pr}\left[X^{r-1}_{i} + W(i,\hat{j}) > X^r ~\land~ X^{r-1}_{i'} + W(i',\hat{j}) > X^r \right] \\
    &= 
        \sum_{i' \in A^{r}\setminus \{i\}} \mathbb{E}\left[(1-\mathrm{e}^{-\lambda(1-X^r+X^{r-1}_{i})})(1-\mathrm{e}^{-\lambda(1-X^r+X^{r-1}_{i'})})\right] \\
    &= 
        O(\lambda^2),
\end{align*}
where we use the fact that $1-\mathrm{e}^{-x} \le x$ for any $x$ for the last relation. 
Thus, summing over all $i \in A^r$, we get
\begin{align*}
    \mathbf{Pr}\big[\hat{j} \in \hat{B}^{r}\big] 
    &= 
        \sum_{i\in A^{r}} \mathbf{Pr}\big[\{i, \hat{j}\} \in {\hat M}\big] \\
    &=
        \sum_{i\in A^{r}} \left(\mathbf{Pr}\big[ X^{r-1}_i + W(i,\hat{j}) > X^r \big] - O(\lambda^2) \right)\\
    &=
        \sum_{i\in A^{r}} \left(\mathbb{E}\left[1-\mathrm{e}^{-\lambda(1-X^r+X^{r-1}_i)}\right]- O(\lambda^2) \right). 
\end{align*}
By Assumption~\ref{assumption:belongs_maximum_weight_matching}, for all $i_1,i_2\in A$, both $\mathbf{Pr}[i_1\in A^r] = \mathbf{Pr}[i_2\in A^r]$ and $\mathbf{Pr}[i_1\notin A^{r-1}] = \mathbf{Pr}[i_2\notin A^{r-1}]$ hold. 
This implies that $\mathbb{E}[X^{r-1}] = \mathbb{E}[X^{r-1}_i]$ since $i$ is selected from $A^{r}$ uniformly at random.
Thus, $\mathbb{E}[X^{r-1}] =\frac{1}{r} \sum_{i'\in A^{r}} \mathbb{E}[X^{r-1}_{i'}]$. Hence, we have
$$
    \lim_{\lambda \to 0}\frac{1}{\lambda}\mathbf{Pr}\left[\hat{j} \in \hat{B}^{r}\right] 
    =
        \sum_{i'\in A^{r}} \mathbb{E}\left[1-X^{r}+X^{r-1}_{i'}\right] 
    =
        r-r(\mathbb{E}[X^{r}] - \mathbb{E}[X^{r-1}]),
$$
which concludes the proof.
\end{proof}



\subsection{Bounds on the Expected Values of $X_p$ and $X_{pq}$}\label{sec:bound}

Next, we establish a lower bound on the expected value of $X_p$, and an upper bound on the expected value of $X_{pq}$.
By utilizing the linearity of expectation, we decompose the expected value into the expected difference accumulated in each round.
%
We then leverage Lemma~\ref{lemma:difference_between_two_expected_values} to analyze the difference between the expected values achieved by each class in two consecutive rounds.

A key observation is that the augmenting path updating the maximum-weight matching in each round can be uniquely determined due to Lemma~\ref{lemma:uniqueness_of_maximum_weight_matching}. 
This uniqueness allows us to identify an edge incident to a newly added vertex in the path and condition on the weight of such an edge. 
Consequently, for the expected value of $X_{p}$, we obtain an upper bound on the limiting probability of the special vertex being included in a maximum-weight matching. Moreover, following a similar proof strategy, we derive an upper bound on the expected value of $X_{pq}$. 

It is important to note that calculating the exact expectation of $X_p$ and $X_{pq}$ proves challenging. While exact computations of the expected minimum total weight of matchings have been explored in random assignment theory, these studies typically assume that the edge weight is drawn from the exponential distribution, whereas in our setting, the edge weight is drawn from the $(\alpha, \beta)$-PDF-bounded distribution that lacks the memorylessness property. Nevertheless, we are able to establish lower and upper bounds on $X_p$ and $X_{pq}$ by carefully applying Lemma~\ref{lemma:difference_between_two_expected_values}.


\subsubsection{Lower bound on $X_p$}
We now present the lower bound on the expected value of $X_p$.
\begin{restatable}{lemma}{ClaimExpexctedXp}\label{claim:expected_diff_p}
    Suppose that $\mathcal{D}$ is $(\alpha,\beta)$-PDF-bounded and condition~\ref{item:assumption-b}. Then, we have
    \begin{align}\label{inequality:expected_diff_p_below}
        \mathbb{E}[X_p] 
        \geq n_p -  \frac{1}{\alpha} \sum_{r=1}^{n_p}\frac{1}{r} \sum_{r'=1}^r  
        \frac{1}{m - r'\cdot k}.
    \end{align}
\end{restatable}

\begin{proof}
By construction of the algorithm, we have $M^{r}(N_p) \subseteq M(N_p)$ and $M^{r-1}(N_p) \cup I^{r}_p \supseteq M(N_p)$. Combining these with Equation~\eqref{eq:greedy} in Proposition~\ref{proposition:1_2_ef1_non_waste} implies that
\[
    M^{r}(N_p) \in \mathrm{argmax} \bigl\{\, v_p(I')\, \bigm|\, \, |I'|=r ~\land~I' \subseteq M(N_p)\,\bigr\}.
\]
Thus, for each $r=1,2,\ldots,n_p$, the maximum-weight matching of size $r$ in $G[N_p,M(N_p)]$ is identical to the maximum-weight matching of size $r$ in $G[N_p,M^{r-1}(N_p) \cup I^{r}_p]$, which is precisely the one created by class $p$ in round $r$ of the round-robin algorithm.

We now proceed to consider the expected value of $X_p=v_p(M(N_p))$. Note that since all edge weights are positive almost surely and $m\geq n$, class $p$ obtains $n_p$ items at the end of the algorithm, i.e., $|M(N_p)|=n_p$. 
For each $r=1,2,\ldots,n_p$, let $X_p^{r}$ denote the maximum-weight of matchings of size $r$ between $N_p$ and $M(N_p)$.
From the previous discussion, $X_p^r = v_p(M^r(N_p))$. 
By linearity of expectation, we can decompose $\mathbb{E}[X_p]$ into the expected marginal utilities of a new item in each round, i.e.,
\begin{equation}\label{eq:linearity_of_expectation_0:main}
    \mathbb{E}[X_p] = \sum_{r=1}^{n_p}\big(\mathbb{E}[X_p^r]-\mathbb{E}[X_p^{r-1}]\big).    
\end{equation}

To apply Lemma~\ref{lemma:difference_between_two_expected_values}, we next add to $G$ a new item $\hat{j}$ and edges from $\hat{j}$ to all $i\in N_p$ with edge-weights independently generated from 
$\mathrm{ReExp}(\lambda)$ on $(-\infty, 1]$. Let ${\hat G}$ denote the resulting graph. 
Let $\hat{B}= M(N_p) \cup \{\hat{j}\}$. 

Let $\hat{M}^r$ denote the maximum-weight matching of size $r$ in ${\hat G}[N_p,\hat{B}]$, and let $\hat{B}^r$ denote the set of items matched to some agent in $N_p$ under $\hat{M}^r$. 

To apply Lemma~\ref{lemma:difference_between_two_expected_values} to $H=G[N_p,M(N_p)]$, let us check whether Assumptions~\ref{assumption:no_distinct} and~\ref{assumption:belongs_maximum_weight_matching} hold in this graph.
By Lemma~\ref{lemma:uniqueness_of_maximum_weight_matching}, Assumption~\ref{assumption:no_distinct} of Lemma~\ref{lemma:difference_between_two_expected_values} is satisfied for $G[N_p,M(N_p)]$. Assumption~\ref{assumption:belongs_maximum_weight_matching} is also satisfied, since by symmetry of distributions, the probability that an agent in $N_p$ is selected until round $r$ by the round-robin algorithm is uniform, which means that the probability that an agent in $N_p$ belongs to a maximum-weight matching of size $r$ in $G[N_p,M(N_p)]$ is uniform across all agents in $N_p$. More formally, consider two agents $i$ and $i'$ in class $p$.
Given the graph $G$, we define $G_{i,i'}$ as the graph obtained by swapping the roles of $i$ and $i'$ in $G$, namely, $i$ (resp. $i'$) and $j$ has an edge weight $w$ in $G$ if and only if $i'$ (resp. $i$) and $j$ has an edge weight $w$ in $G_{i,i'}$. Now, let $Y$ denote a random variable for $G$ and $Z$ denote a random variable such that $Y=G$ if and only if $Z=G_{i,i'}$. 
Let $\mathcal{A}_i$ be the event that agent $i$ appears in a maximum-weight matching of size $r$ in $G[N_p,M(N_p)]$ where $G$ is generated from the distribution of $Y$. Similarly, let $\mathcal{B}_{i}$ (resp. $\mathcal{B}_{i'}$) be the corresponding event for $Z$. 
Then, we have 
$\mathbf{Pr}[\mathcal{A}_i] = \mathbf{Pr}[\mathcal{B}_{i'}] = \mathbf{Pr}[\mathcal{A}_{i'}]$, where the first equality holds by construction of $Z$ and the second equality holds by the fact that in both $G$ and $G_{i,i'}$, the weights of edges are drawn independently from a distribution $\mathcal{D}$. 
Therefore, for any agent, the probability of being included in a maximum-weight matching of fixed size $r$ in $G[N_p,M(N_p)]$ is the same.
Thus, we can apply Lemma~\ref{lemma:difference_between_two_expected_values} to $H=G[N_p,M(N_p)]$ and get
\begin{equation}\label{eq:difference_between_two_B_r_0:main}
    \mathbb{E}[X_p^r]-\mathbb{E}[X_p^{r-1}] 
    = 1- \frac{1}{r}\lim_{\lambda \to 0}\frac{1}{\lambda}\mathbf{Pr}\left[\hat{j} \in \hat{B}^{r}\right].
\end{equation}

Next, we present an upper bound on the limited probability $\lim_{\lambda \to 0} \frac{1}{\lambda}\mathbf{Pr}[ \hat{j} \in \hat{B}^{r}]$. 
We claim that 
\begin{equation}\label{inequality:hat_j_bound:main}
    \frac{1}{r}\lim_{\lambda \to 0}\frac{1}{\lambda} \mathbf{Pr}\left[\hat{j} \in \hat{B}^{r}\right] 
    \leq \frac{1}{\alpha } \sum_{r=1}^{n_p}\frac{1}{r} \sum_{r'=1}^r  
        \frac{1}{m - r' \cdot k}.
\end{equation}
Inequality~\eqref{inequality:expected_diff_p_below} in Lemma~\ref{claim:expected_diff_p} can be established by combining Equations~\eqref{eq:linearity_of_expectation_0:main},~\eqref{eq:difference_between_two_B_r_0:main} and Inequality~\eqref{inequality:hat_j_bound:main}.

To show Inequality~\eqref{inequality:hat_j_bound:main}, we calculate the conditional probability $\mathbf{Pr}[\hat{j} \in \hat{B}^{r} \bigm| \hat{j} \notin \hat{B}^{r-1} ]$. Namely, we consider  the probability that $\hat{j}$ appears in matching $\hat{M}^r$ under the condition that $\hat{j}$ does not appear in $\hat{M}^{r-1}$. Let us assume that $\hat{j}$ does not appear in $\hat{M}^{r-1}$. Then $\hat{M}^{r-1} = M^{r-1}$ since both represent the maximum-weight matching of size $r-1$ between $N_p$ and $M(N_p)$ in ${\hat G}$. Further, $\hat{M}^r$ is a maximum-weight matching of size $r$ in the graph ${\hat G}[N_p,{M}^{r-1}(N_p) \cup I^r_p \cup \{\hat j\}]$ since the total weight of $\hat{M}^r$ is at least as large as that of ${M}^r$ and ${M}^r$ is a maximum-weight matching of size $r$ in the graph $G[N_p,M^{r-1}(N_p) \cup I^{r}_p]$. 

Consider an augmenting path $P$ of $M^{r-1}$ in the bipartite graph ${\hat G}[N_p,M^{r-1}(N_p) \cup I^r_p \cup \{\hat j\}]$ that yields $\hat{M}^r$. This path is uniquely determined since $\hat{M}^r$ is uniquely determined almost surely by Lemma~\ref{lemma:uniqueness_of_maximum_weight_matching}. Suppose that $P$ starts from an agent $i^r$ and ends with a new item $\hat{j}^{r}$. 
Then, we have 
$$
\mathbf{Pr}\left[\hat{j} \in \hat{B}^{r} \bigm| \hat{j} \notin \hat{B}^{r-1}  \right] = \mathbf{Pr}\left[\hat{j}^{r} = \hat{j}  \bigm| \hat{j} \notin \hat{B}^{r-1}  \right].
$$

For $r=1,2,\ldots,n_p$, let $W^{(r)}$ be the random variable representing the weight of the edge which is incident to $\hat{j}^{r}$ in the augmenting path.
Let $W(i,j)$ denote the random variable representing the weight of the edge $\{i,j\}$ for each $i\in N_p$ and $j \in M(N_p) \cup \{\hat{j}\}$.
Let $i_{\ell}$ represent the agent in $N_p$ that is adjacent to $\hat{j}^{r}$ in the augmenting path. 
We condition on a new agent who participates in a maximum-weight matching of size $r$ together with weights of the edges involving that agent. 
Specially, we consider the following three conditions: 
\begin{itemize}
\item[$($i$)$] the agent $i_{\ell}$ matched to a new item $\hat{j}^{r}$ in the augmenting path is fixed, 
\item[$($ii$)$] the weights of all edges other than the edges $\{i_{\ell},j\}$ $(j \in I^{r}_p\cup \{\hat{j}\})$ are fixed, and 
\item[$($iii$)$] $W^{(r)} = w^{(r)}$, which implies that the weights $W(i_{\ell},j)$ for all $j \in I^{r}_p\cup \{\hat{j}\} $ are upper bounded by $w^{(r)}$, i.e., $\max \{W(i_{\ell},j) \mid j \in I^{r}_p\cup \{\hat{j}\}\} \leq w^{(r)}$. 
\end{itemize}
If we condition on $($i$)$, $($ii$)$, and $($iii$)$, then the probability that $\hat{j}^{r} = \hat{j}$ equals the probability that $W(i_{\ell},\hat{j})$ is the maximizer among $\{W(i_{\ell},j) \mid j \in I^{r}_p\cup \{\hat{j}\} \}$.
Thus, we obtain
\begin{align*}
    &\mathbf{Pr}\left[ \hat{j} \in \hat{B}^{r} \Bigm| 
    \hat{j}\notin \hat{B}^{r-1} ~\land~ \text{(i), (ii) and (iii)} \right]  \nonumber \\
    &=\mathbf{Pr}\left[ w(i_{\ell},\hat{j}) = \max \{ W(i_{\ell},j) \mid j \in I^{r}_p\cup \{\hat{j}\} \} \Bigm| 
    \hat{j}\notin \hat{B}^{r-1} ~\land~ \text{(i), (ii) and (iii)} \right] .
\end{align*}
Let $m^{r'} =|I^{r'}_p|$ for each $r'=1,2,\ldots,r$.
For each $i\in N_p$ and each $j\in I^r_p$, $W(i,j)$ is drawn from $\mathcal{D}$, which is $(\alpha,\beta)$-PDF bounded, and for each $i\in N_p$, $W(i,\hat{j})$ is drawn from the reversed exponential distribution $\mathrm{ReExp}(\lambda)$.
For the sake of simplicity, we will refer to conditions (i), (ii), and (iii) collectively as $\mathcal{E}_{w^{(r)}}$.
Then, we get
\begin{align*}
    \mathbf{Pr}\left[ \max \{W(i_{\ell},j) \mid j \in I^{r}_p\cup \{\hat{j}\}\} \leq w^{(r)} \Bigm| \mathcal{E}_{w^{(r)}} \right]  
    &= \int_{-\infty}^{w^{(r)}} \lambda \mathrm{e}^{-\lambda (1-x)}  \mathrm{d}x \cdot F_{\mathcal{D}}(w^{(r)})^{m^r} \\
    &= \mathrm{e}^{-\lambda (1-w^{(r)})}  F_{\mathcal{D}}(w^{(r)})^{m^r},
\end{align*}
and
\begin{align*}
    &\mathbf{Pr}\left[ W(i_{\ell}, \hat{j}) = \max \{ W(i_{\ell},j) \mid j \in I^{r}_p\cup \{\hat{j}\}\} ~\land~ 
    \max \{W(i_{\ell},j) \mid j \in I^{r}_p\cup \{\hat{j}\}\}  \leq w^{(r)} \Bigm| \mathcal{E}_{w^{(r)}}\right] \\
    & = \int_{0}^{w^{(r)}} 
    \mathbf{Pr}\left[ W(i_{\ell}, \hat{j}) =  x\right] \cdot
    \mathbf{Pr}\left[ W(i_{\ell}, j) \leq  x~\mbox{for every }j \in I^r_p\right] 
    \mathrm{d}x \\
    &= \int_{0}^{w^{(r)}} \lambda \mathrm{e}^{-\lambda (1-x)}  F_{\mathcal{D}}(x)^{m^r}\mathrm{d}x.
\end{align*}
From the above two equations, we obtain
\begin{align*}
    &\mathbf{Pr}\left[ \hat{j} \in \hat{B}^{r} \Bigm| \hat{j} \notin \hat{B}^{r-1}~\land~ \mathcal{E}_{w^{(r)}} \right] \\
    &= \frac{\mathbf{Pr}\left[ W(i_{\ell}, \hat{j}) = \max \{ W(i_{\ell},j) \mid j \in I^{r}_p\cup \{\hat{j}\}\} \}\ ~\land~ \max \{W(i_{\ell},j) \mid j \in I^{r}_p\cup \{\hat{j}\}\} \} \leq w^{(r)} \Bigm| \mathcal{E}_{w^{(r)}}\right]}
    {\mathbf{Pr}\left[ \max \{W(i_{\ell},j) \mid j \in I^{r}_p\cup \{\hat{j}\}\} \leq w^{(r)} \Bigm| \mathcal{E}_{w^{(r)}}\right]} \\
    &= \lambda \mathrm{e}^{\lambda (1-w^{(r)})} \int_{0}^{w^{(r)}} \mathrm{e}^{-\lambda (1-x)} \left( \frac{F_{\mathcal{D}}(x)}{F_{\mathcal{D}}(w^{(r)})}\right)^{m^r} \mathrm{d}x.
\end{align*}
This implies that
\begin{align*}
    \mathbf{Pr}\left[\hat{j} \in \hat{B}^{r} \Bigm| \mathcal{E}_{w^{(1)}},\ldots,\mathcal{E}_{w^{(r)}} \right] 
    &= 1 - 
    \prod_{r'=1}^r \mathbf{Pr}\left[ \hat{j} \notin \hat{B}^{r'} \Bigm| \hat{j} \notin \hat{B}^{r'-1}~\land~ \mathcal{E}_{w^{(r')}} \right]  \\ 
    &= 1 - 
    \prod_{r'=1}^r \left(1- \mathbf{Pr}\left[ \hat{j} \in \hat{B}^{r'} \Bigm| \hat{j} \notin \hat{B}^{r'-1} ~\land~ \mathcal{E}_{w^{(r')}} \right]\right)  \\ 
    &= 1 -  \prod_{r'=1}^r \left( 1 - \lambda  \mathrm{e}^{\lambda  (1-w^{(r')})} \int_{0}^{w^{(r')}} \mathrm{e}^{-\lambda (1-x)}  \left( \frac{F_{\mathcal{D}}(x)}{F_{\mathcal{D}}(w^{(r')})}\right)^{m^{r'}} \mathrm{d}x \right) \\  
    &=\sum_{r'=1}^r \lambda  \mathrm{e}^{\lambda   (1-w^{(r')})} \int_{0}^{w^{(r')}} \mathrm{e}^{-\lambda (1-x)}  \left( \frac{F_{\mathcal{D}}(x)}{F_{\mathcal{D}}(w^{(r')})}\right)^{m^{r'}} \mathrm{d}x + O_{\lambda}(\lambda^2).
\end{align*}
Furthermore, from uniform convergence, we achieve
\begin{align*}\label{eq:limit_of_prob}
    \lim_{\lambda \to 0} \frac{1}{\lambda}\mathbf{Pr} \left[\hat{j} \in \hat{B}^{r} \bigm| \mathcal{E}_{w^{(1)}},\ldots,\mathcal{E}_{w^{(r)}} \right] 
    &= \lim_{\lambda \to 0} \sum_{r'=1}^r \mathrm{e}^{\lambda  (1-w^{(r')})} \int_{0}^{w^{(r')}} \mathrm{e}^{-\lambda (1-x)} \left( \frac{F_{\mathcal{D}}(x)}{F_{\mathcal{D}}(w^{(r')})}\right)^{m^{r'}} \mathrm{d}x  \\
    &=\sum_{r'=1}^r \int_{0}^{w^{(r')}}  \left( \frac{F_{\mathcal{D}}(x)}{F_{\mathcal{D}}(w^{(r')})}\right)^{m^{r'}} \mathrm{d}x.
\end{align*}
Set $y = \frac{F_{\mathcal{D}}(x)}{F_{\mathcal{D}}(w^{(r')})}$. Here we have $\mathrm{d}y = \frac{f_{\mathcal{D}}(x)}{F_{\mathcal{D}}(w^{(r')})}\, \mathrm{d}x$. Since $\mathcal{D}$ is $(\alpha,\beta)$-PDF-bounded, we have $f_{\mathcal{D}}(x) \geq \alpha$. Also, we have $F_{\mathcal{D}}(x) \leq 1$ for all $x \in [0,1]$. Hence, for all $r'=1,2,\ldots,r$, we have
\begin{align*}
    \int_{0}^{w^{(r')}}  \left( \frac{F_{\mathcal{D}}(x)}{F_{\mathcal{D}}(w^{(r')})}\right)^{m^{r'}} \mathrm{d}x
    =\int_{0}^{1}  y^{m^{r'}}\frac{F_{\mathcal{D}}(w^{(r')})}{f_{\mathcal{D}}(x)}\mathrm{d}y 
    \leq \frac{1}{\alpha}\int_{0}^{1}  y^{m^{r'}}\mathrm{d}y 
    =\frac{1}{\alpha} \frac{1}{m^{r'}+1}.
\end{align*}
The right hand side of this inequality does not depend on $w^{(1)},w^{(2)},\ldots,w^{(r)}$. Thus, we can remove the conditioning about $w^{(1)},w^{(2)},\ldots,w^{(r)}$ and then we obtain
\begin{equation}\label{eq:upper_bound_j_in_B_r}
    \lim_{\lambda \to 0} \frac{1}{\lambda}\mathbf{Pr}\left[ \hat{j} \in \hat{B}^{r}\right] 
    \leq \sum_{r'=1}^r \frac{1}{\alpha} \frac{1}{m^{r'}+1}.
\end{equation}

Finally, we consider the lower bound on the number $m^{r'}$ of remaining items when class $p$ selects a new item in round $r'$. 
The number of items already taken just before class $p$ chooses an item in round $r'$ is given by 
\[
    m^{r'} = |I^{r'}_p| = m - \sum_{p'= 1}^k \min(r'-1,n_{p'}) - \sum_{p'= 1}^{p-1} 1\cdot \mathbbm{1}(r' \le n_{p'}).
\]
This implies that $m^{r'} = |I^{r'}_p| \geq m - (r'-1)\cdot k - (p-1)$. 
Since $k\geq p-2$, we obtain
\begin{align}\label{eq:m_r}
    m^{r'}+1
    &\ge m - (r'-1)\cdot k - (p-1) + 1 \nonumber\\
    &\ge m -(r' -1) \cdot  k  -p+2 \nonumber\\
    &\ge m - r' \cdot k
\end{align}
for each $r'=1,2,\ldots,n_p$. 
From Inequalities~\eqref{eq:upper_bound_j_in_B_r} and~\eqref{eq:m_r}, we obtain Inequality~\eqref{inequality:hat_j_bound:main}.
\end{proof}


\subsubsection{Upper bound on $X_{pq}$}

We next turn our attention to the upper bound on the expected value of $X_{pq} = v_p(M(N_q))$. 
Since there are almost surely no edges with zero weight between $N_p$ and $M(N_q)$, the size of the maximum-weight matching between $N_p$ and $M(N_q)$ is $\min(n_p,n_q)$ almost surely. Also, by Lemma~\ref{lemma:uniqueness_of_maximum_weight_matching},
the maximum-weight matching of a fixed size between $N_p$ and $M(N_q)$ is uniquely determined. 
Now, for each $r=1,2,\ldots,\min(n_p,n_q)$,
let $j^{r}_q$ denote the item which appears in the maximum-weight matching of size $r$ but does not appear in that of size $r-1$. Here, $j^{r}_q$ is also uniquely determined from Lemma~\ref{lemma:nesting_lemma}.
In addition, let $W^{(r)}_q$ be the random variable representing the weight of the edge that is adjacent to $j^{r}_q$ and included in the maximum-weight matching of size $r$ between $N_p$ and $M(N_q)$.

With these in hand, we present the upper bound on the expected value of $X_{pq}$. 
\begin{restatable}{lemma}{ClaimExpexctedq}\label{claim:expected_diff_q}
    Suppose that $\mathcal{D}$ is $(\alpha,\beta)$-PDF-bounded and $m\ge n$. 
    Then, we have
    \begin{align*}
        \mathbb{E}\big[X_{pq}\big]
        \leq \min(n_p,n_q)-   
        \frac{\alpha}{\beta} \sum_{r=1}^{\min(n_p,n_q)} 
        \frac{1}{r}\sum_{r'=1}^r 
        \frac{\mathbb{E}\big[W^{(r')}_q\big]}{n_q-r'+2}.  
    \end{align*}
\end{restatable}
A natural way to obtain the upper bound is to estimate the expected value $\mathbb{E}[X_{pq}]$ directly. However, it is difficult to do this since we do not know how class $p$ evaluates the marginal gain class $q$ enjoys at each round. 
Instead, we consider a bundle $B_u$ of size $n_q$, which is selected uniformly at random from $I$. By conditioning on $M(N_q)=B_u$ and analyzing the maximum matching between $N_p$ and $B_u$, we circumvent the difficulty and apply a similar argument to that for Lemma~\ref{claim:expected_diff_p} to establish Lemma~\ref{claim:expected_diff_q}.

\begin{proof}[Proof of Lemma~\ref{claim:expected_diff_q}]
Since we assume $m\ge n$, we have $|M(N_q)|=n_q$.
Let $X_{pq}^{r}$ denote the maximum weight of matchings of size $r$ between $N_p$ and $M(N_q)$. 
We introduce a new item, $\hat{j}$, to $M(N_q)$ and connect it to every agent in $N_p$ by an edge whose weight is generated independently from the reversed exponential distribution $\mathrm{ReExp}(\lambda)$.
Let $\hat{G}$ denote the resulting graph.
Let $\hat{B}^r_{pq}$ be the set of items that appears in the maximum-weight matching between $N_p$ and $M(N_q) \cup \{\hat{j}\}$ in $\hat{G}$.

To apply Lemma~\ref{lemma:difference_between_two_expected_values} to $H=G[N_p,M(N_q)]$, we verify that Assumptions~\ref{assumption:no_distinct} and~\ref{assumption:belongs_maximum_weight_matching} hold for this graph.
Lemma~\ref{lemma:uniqueness_of_maximum_weight_matching} ensures that Assumption~\ref{assumption:no_distinct} of Lemma~\ref{lemma:difference_between_two_expected_values} is satisfied for $G[N_p,M(N_q)]$. 
Assumption~\ref{assumption:belongs_maximum_weight_matching} is also satisfied by a similar argument to the one in the proof of \ref{claim:expected_diff_p}, namely, for any two agents in $N_p$, the probability of each agent being included in a maximum-weight matching of fixed size $r$ in $G[N_p,M(N_q)]$ is the same. 

Finally, from Lemma~\ref{lemma:difference_between_two_expected_values} we obtain that
\begin{align}\label{eq:diff_X_pq_r}
    \mathbb{E}\big[X_{pq}^r \big] - \mathbb{E}\big[X_{pq}^{r-1}\big] = 1 - \frac{1}{r}\lim_{\lambda \to 0}\frac{1}{\lambda} \mathbf{Pr}\big[ \hat{j} \in \hat{B}^{r}_{pq} \big].
\end{align}
Thus, our goal is to establish a lower bound for the limiting probability $\lim_{\lambda \to 0}\frac{1}{\lambda} \mathbf{Pr}\big[ \hat{j} \in \hat{B}^{r}_{pq} \big]$.

Analyzing $\mathbf{Pr}\big[\hat{j} \in \hat{B}^{r}_{pq} \big]$ directly is challenging since the only information we have is that the items in $M(N_q)$ are not selected by class $p$ but we do not know how class $p$ evaluates the marginal gain class $q$ enjoys from $p$'s viewpoint over the course of the round-robin algorithm. 

Now, we select a subset of size $n_q$ from $I$ uniformly at random, denoted by $B_u \subseteq I$. 
Let $M_u^r$ denote the maximum-weight matching of size $r$ between $N_p$ and $B_u$ in $\hat{G}$. 
Let $B^r_u$ denote the set of items which appear in matching $M_u^r$. By Lemma~\ref{lemma:nesting_lemma}, $B^1_u \subseteq B^2_u \subseteq \cdots \subseteq B^{\min(n_p,n_q)}_u$. 
Let $I_u^r = B_u \setminus B_u^{r-1}$ and $\hat{B}_u= B_u \cup \{\hat{j}\}$.
For $r = 1,2,\ldots,\min(n_p,n_q)$, let $\hat{M}_u^r$ denote the maximum-weight matching of size $r$ between $N_p$ and $\hat{B}_u$ in $\hat{G}$, and let $\hat{B}^r_u$ denote the set of items which appear in matching $\hat{M}^r_u$. By symmetry, for all $S\subseteq I$ such that $|S|=n_q$, the probability that $S$ coincides with $M(N_q)$ is the same.
Here, we have
\begin{equation}\label{eq:B_r_u}
    \mathbf{Pr}\big[\hat{j} \in \hat{B}^{r}_{pq} \big]
    = \sum_{B_u\in I[n_q]} 
    \mathbf{Pr}\big[\hat{j} \in \hat{B}^{r}_u \bigm| 
    B_u = M(N_q) \big] 
    \cdot \mathbf{Pr}[B_u = M(N_q)],
\end{equation}
where $I[n_q]$ denotes the family of subsets of size $n_q$ of $I$.

We now proceed to consider the discussion found in the proof of Inequality~\eqref{inequality:hat_j_bound:main}. 
We begin by considering the conditional probability
$
\mathbf{Pr}\left[\hat{j} \in \hat{B}^{r}_u  \bigm| \hat{j} \notin \hat{B}^{r-1}_u  ~\land~ B_u=M(N_q) \right].
$

Let us assume that $\hat{j} \notin \hat{B}_u^{r-1}$ and $B_u=M(N_q)$. Then, we have $\hat{M}_u^{r-1} = M_u^{r-1}$. 
Again, consider an augmenting path of matching $\hat{M}_u^{r-1}$ in the modified graph $\hat{G}$ restricted to $N_p$ and $B_u\cup \{\hat{j}\}$ that results in a maximum-weight matching of size $r$ between $N_p$ and $B_u\cup \{\hat{j}\}$. 
Abuse of notation, let $\hat{j}^{r} \in I_u\cup \{\hat{j}\}$ denote the new item of the augmenting path that does not appear in $\hat{B}_u^{r-1}$. We now consider the probability that $\hat{j}^{r} = \hat{j}$ similarly to the proof of Lemma~\ref{claim:expected_diff_p}.
Let $i_{\ell}$ denote the vertex adjacent to $\hat{j}^{r}$ in the augmenting path. 
Let $W_u^{(r)}$ denote the random variable representing the weight of the edge $\{i_{\ell}, \hat{j}^{r}\}$.
Let $i^{r} \in N_p$ represent the left endpoint of the alternating path. 
The augmenting path begins with $i^{r}$ and 
ends with $\hat{j}^{r}$.

We now consider the following three conditions: 
\begin{itemize}
\item[$($i$)$] the agent $i_{\ell}$ matched to a new item $\hat{j}^{r}$ in the augmenting path is fixed, 
\item[$($ii$)$] the weights of all edges other than the edges $\{i_{\ell},j\}$ $(j \in I_u^{r}\cup \{\hat{j}\})$ are fixed, and 
\item[$($iii$)$] $ W_u^{(r)} = w^{(r)}_u$, which implies that the weights $W(i_{\ell},j)$ for all $j \in I^{r}_u\cup \{\hat{j}\} $ are upper bounded by $w^{(r)}_u$, i.e., $\max \{W(i_{\ell},j) \mid j \in I^{r}_u\cup \{\hat{j}\}\} \leq w^{(r)}_u$. 
\end{itemize}
Here, we obtain
\begin{align}\label{eq:condioning_eqation_2_q}
    &\mathbf{Pr}\left[\hat{j} \in \hat{B}_u^{r} \Bigm| 
    \hat{j}\notin \hat{B}_u^{r-1} ~\land~\text{(i), (ii) and (iii)}~\land~B_u=M(N_q)  \right]  \nonumber \\
    &=\mathbf{Pr}\left[ W(i_{\ell}, \hat{j}) = \max \{ W(i_{\ell},j) \mid j \in I_u^{r}\cup \{\hat{j}\}\} \Bigm| 
    \hat{j}\notin \hat{B}_u^{r-1}~\land~\text{(i), (ii) and (iii)}~\land~B_u=M(N_q) \right].
\end{align}

Next, let $\mathcal{S}$ represent the event that $W(i_{\ell}, \hat{j}) = \max \{ W(i_{\ell},j) \mid j \in I_u^{r}\cup \{\hat{j}\}\}$, 
$\mathcal{T}$ denote the event that $\hat{j} \notin \hat{B}^{r-1}_u$ and (i), (ii) and (iii), 
and $\mathcal{U}$ denote the event that $B_u = M(N_q)$.
The right-hand side of 
Equation~\eqref{eq:condioning_eqation_2_q} 
can be expressed as
$\mathbf{Pr}[\mathcal{S} \mid (\mathcal{T}\land \mathcal{U})].
$ 
We aim to show that $\mathbf{Pr}[\mathcal{U} \mid (\mathcal{S}\land \mathcal{T})] \geq \mathbf{Pr}[\mathcal{U} \mid \mathcal{T}]$. 
Conditioning on $\mathcal{S}$ implies that, for the subset $I_u^r\subseteq B_u$ and for all $j$ in the subset, the inequality $W(i_{\ell},j) \leq W(i_{\ell},\hat{j})$ holds.
This means that for all $j \in I^r_u$, $W(i_{\ell}, j)$ is upper bounded by the value $W(i_{\ell}, \hat{j})$. 
Thus, when $\mathcal{S}$ is conditioned,
the probability of an item in $I_u^r$ being chosen by agent $i_{\ell}$ in the round-robin algorithm is lower. Also, if the probability of an item in $I^r_u$ being chosen by some agent of class $p$ through the round-robin algorithm decreases, the likelihood of such an item being selected by other classes increases.
Moreover, the presence of the condition $\mathcal{T}$ does not affect the probability of an item in $I_u^r$ being chosen by agent $i_{\ell}$. 
Thus, the probability that some items in $B_u$ are selected by class $q$ under the condition $\mathcal{S}\land \mathcal{T}$ is greater than that under $\mathcal{T}$ only.
Therefore, it follows that $\mathbf{Pr}[\mathcal{U} \mid (\mathcal{S}\wedge \mathcal{T})] \geq \mathbf{Pr}[\mathcal{U} \mid \mathcal{T}]$. By applying Bayes' theorem, we derive $\mathbf{Pr}[\mathcal{S} \mid (\mathcal{T}\land \mathcal{U})] \geq \mathbf{Pr}[\mathcal{S} \mid \mathcal{T}]$. 
In other words, the events $\mathcal{S}$ and $\mathcal{U}$ are positively associated when conditioned on $\mathcal{T}$. From this, we obtain
\begin{align}\label{ineq:condioning_equation_3}
    &\mathbf{Pr}\left[ W(i_{\ell}, \hat{j}) = \max \{ W(i_{\ell},j) \mid j \in I_u^{r}\cup \{\hat{j}\}\} \Bigm| 
    \hat{j}\notin \hat{B}_u^{r-1}~\land~\text{(i), (ii) and (iii)}~\land~B_u=M(N_q) \right]  \nonumber\\
    &\geq 
    \mathbf{Pr}\left[W(i_{\ell}, \hat{j}) = \max \{ W(i_{\ell},j) \mid j \in I_u^{r}\cup \{\hat{j}\}\}  \Bigm| 
    \hat{j}\notin \hat{B}_u^{r-1}~\land~\text{(i), (ii) and (iii)} \right].
\end{align}

Recall that $|I_u^{r}|=n_q-(r-1) = n_q-r+1$.
Similarly to the event $\mathcal{E}_{w^{(r)}}$ defined previously, we here refer to conditions (i), (ii), and (iii) collectively as $\mathcal{E}_{w_u^{(r)}}$.
Then, we have
\begin{align}\label{eq:conditional_prob_4}
    &\mathbf{Pr}\left[W(i_{\ell}, \hat{j}) = \max \{ W(i_{\ell},j) \mid j \in I_u^{r}\cup \{\hat{j}\}\}  \Bigm| 
    \hat{j}\notin \hat{B}_u^{r-1}~\land~\mathcal{E}_{w_u^{(r)}} \right] \nonumber \\
    &= \frac{\mathbf{Pr}\left[ W(i_{\ell}, \hat{j}) = \max \{ W(i_{\ell},j) \mid j \in I_u^{r}\cup \{\hat{j}\}\} \}\ ~\land~ \max \{W(i_{\ell},j) \mid j \in I_u^{r}\cup \{\hat{j}\}\} \} \leq w_u^{(r)} \Bigm| \mathcal{E}_{w_u^{(r)}} \right]}
    {\mathbf{Pr}\left[ \max \{W(i_{\ell},j) \mid j \in I_u^{r}\cup \{\hat{j}\}\} \leq w_u^{(r)} \Bigm| \mathcal{E}_{w_u^{(r)}} \right]} \nonumber\\
    &= \lambda \mathrm{e}^{\lambda (1-w_u^{(r)})} \int_{0}^{w_u^{(r)}} \mathrm{e}^{-\lambda (1-x)} \left( \frac{F_{\mathcal{D}}(x)}{F_{\mathcal{D}}(w_u^{(r)})}\right)^{n_q-r+1} \mathrm{d}x.
\end{align}
From Equation~\eqref{eq:condioning_eqation_2_q}, Inequality~\eqref{ineq:condioning_equation_3} and Equation~\eqref{eq:conditional_prob_4}, we obtain 
\begin{align*}
    &\mathbf{Pr}\left[\hat{j} \in \hat{B}_u^{r} \Bigm| 
    \hat{j}\notin \hat{B}_u^{r-1} ~\land~\mathcal{E}_{w_u^{(r)}}~\land~B_u=M(N_q)  \right]  \\
    &\geq 
    \lambda \mathrm{e}^{\lambda (1-w_u^{(r)})} \int_{0}^{w_u^{(r)}} \mathrm{e}^{-\lambda (1-x)} \left( \frac{F_{\mathcal{D}}(x)}{F_{\mathcal{D}}(w_u^{(r)})}\right)^{n_q-r+1} \mathrm{d}x.
\end{align*}
This implies that
\begin{align*}
    &\mathbf{Pr}\left[\hat{j} \in \hat{B}^{r} \Bigm| \mathcal{E}_{w_u^{(1)}}\land\cdots \land\mathcal{E}_{w_u^{(\min(n_p,n_q))}}~\land~ B_u=M(N_q)  \right] \\
    &= 1 - 
    \prod_{r'=1}^r \mathbf{Pr}\left[ \hat{j} \notin \hat{B}^{r'} \Bigm| \hat{j} \notin \hat{B}^{r'-1}~\land~\mathcal{E}_{w_u^{(1)}}\land\cdots \land\mathcal{E}_{w_u^{(\min(n_p,n_q))}}~\land~ B_u=M(N_q) \right]  \\ 
    &\geq 1 -  \prod_{r'=1}^r \left( 1 - \lambda  \mathrm{e}^{\lambda  (1-w_u^{(r')})} \int_{0}^{w_u^{(r')}} \mathrm{e}^{-\lambda (1-x)}  \left( \frac{F_{\mathcal{D}}(x)}{F_{\mathcal{D}}(w_u^{(r')})}\right)^{n_q-r'+1} \mathrm{d}x \right) \\  
    &=\sum_{r'=1}^r \lambda  \mathrm{e}^{\lambda   (1-w_u^{(r')})} \int_{0}^{w_u^{(r')}} \mathrm{e}^{-\lambda (1-x)}  \left( \frac{F_{\mathcal{D}}(x)}{F_{\mathcal{D}}(w_u^{(r')})}\right)^{n_q-r'+1} \mathrm{d}x + O_{\lambda}(\lambda^2).
\end{align*}
Thus, we obtain
\begin{align*}
    &\lim_{\lambda \to 0} \frac{1}{\lambda} \mathbf{Pr}\left[\hat{j} \in \hat{B}^{r} \Bigm| \mathcal{E}_{w_u^{(1)}}\land\cdots \land\mathcal{E}_{w_u^{(\min(n_p,n_q))}}~\land~ B_u=M(N_q)  \right] \\
    &\ge \sum_{r'=1}^r \int_{0}^{w_u^{(r')}} \left( \frac{F_{\mathcal{D}}(x)}{F_{\mathcal{D}}(w_u^{(r')})}\right)^{n_q-r'+1} \mathrm{d}x.
\end{align*}
Set $y = \frac{F_{\mathcal{D}}(x)}{F_{\mathcal{D}}(w_u^{(r)})}$. This implies $\mathrm{d}y = \frac{f_{\mathcal{D}}(x)}{F_{\mathcal{D}}(w^{(r)})} \mathrm{d}x$. 
Since $\mathcal{D}$ is $(\alpha,\beta)$-PDF-bounded, $f_{\mathcal{D}}(x) \leq \beta$. 
Moreover, we have $F_{\mathcal{D}}(w_u^{(r')}) \geq \alpha w_u^{(r')}$. Then, 
\begin{align*}
    &\lim_{\lambda \to 0} \frac{1}{\lambda} \mathbf{Pr}\left[\hat{j} \in \hat{B}^{r} \Bigm| \mathcal{E}_{w_u^{(1)}}\land\cdots \land\mathcal{E}_{w_u^{(\min(n_p,n_q))}}~\land~ B_u=M(N_q) \right] \\
    &\geq \sum_{r'=1}^r \int_{0}^{1} y^{n_q-r'+1} \frac{F_{\mathcal{D}}(w_u^{(r')})}{f_{\mathcal{D}}(x)} \mathrm{d}y \\
    &\geq \frac{\alpha}{\beta}  \sum_{r'=1}^r w_u^{(r')}\int_{0}^{1} y^{n_q-r'+1} \mathrm{d}y \\
    &\geq \frac{\alpha}{\beta} \sum_{r'=1}^r  \frac{w_u^{(r')}}{n_q-r'+2}.
\end{align*}
Calculating the expected values for both sides concerning $w^{(1)},w^{(2)},\ldots,w^{(r)}$, we achieve 
\[
    \lim_{\lambda \to 0} \frac{1}{\lambda} \mathbf{Pr}\left[\hat{j} \in \hat{B}^{r} \Bigm| B_u=M(N_q) \right] 
    \geq \frac{\alpha}{\beta} \sum_{r'=1}^r  \frac{\mathbb{E}\bigl[W_u^{(r')} \bigm| B_u=M(N_q)  \bigr]}{n_q-r'+2}.
\]
This inequality, Equation~\eqref{eq:diff_X_pq_r}, Equation~\eqref{eq:B_r_u} imply that
\begin{align}\label{inequality:w_q_r_}
    \lim_{\lambda \to 0}\frac{1}{\lambda}\mathbf{Pr}\big[\hat{j} \in \hat{B}^{r}_{pq} \big]
    &\geq \sum_{B_u\in I[n_q]} 
    \lim_{\lambda \to 0}\frac{1}{\lambda}\mathbf{Pr}\big[\hat{j} \in \hat{B}^{r}_u \bigm| 
    B_u = M(N_q) \big] 
    \cdot \mathbf{Pr}[B_u = M(N_q)] \nonumber \\
    &\geq \frac{\alpha}{\beta} \sum_{r'=1}^r  \frac{\mathbb{E}\bigl[W_q^{(r')} \bigr]}{n_q-r'+2}.
\end{align}
By combining Equation~\eqref{eq:diff_X_pq_r} and Inequality~\eqref{inequality:w_q_r_}, we establish Lemma~\ref{claim:expected_diff_q}.
\end{proof}


\subsection{No Light Edges in Maximum-Weight Matchings}\label{sec:no_heavy_edges}
From Lemmas~\ref{claim:expected_diff_p} and~\ref{claim:expected_diff_q}, we can obtain a lower bound for $\mathbb{E}[X_{p}] - \mathbb{E}[X_{pq}]$.
To achieve probabilistic concentrations around expected values, we demonstrate that the edge weights in the maximum-weight matching between $N_p$ and $M(N_p)$ and those in the maximum-weight matching between $N_p$ and $M(N_q)$ are sufficiently heavy. This is formalized in Lemma~\ref{lem:no_heavy_edge}.
\begin{restatable}{lemma}{LemNoLightEdges}\label{lem:no_heavy_edge}
    Suppose that distribution $\mathcal{D}$ is $(\alpha,\beta)$-PDF-bounded and condition~\ref{item:assumption-c}. Then, we have the followings:
    \begin{itemize}
        \item[$(\mathrm{a})$] 
        Let $c_p$ be a constant such that $c_p > 7000\alpha^{-1} > 0$. For all $r=1,2,\ldots,n_p$, every edge in the maximum-weight matching of size $r$ between $N_p$ and $M(N_p)$ has a weight at least $1-c_p \frac{(\log n_p)^2}{n_p}$ with a probability of at least $1 - O(n_p^{-3})$. 
        \item[$(\mathrm{b})$] 
        Let $c_q$ be a constant such that $c_q > 60c_p + 400\alpha^{-1} >0 $. For all $r=1,2,\ldots,\min(n_p,n_q)$, every edge in the maximum-weight matching of size $r$ between $N_p$ and $M(N_q)$ has a weight at least $1-c_q \frac{(\log \min(n_p,n_q))^4}{\min(n_p,n_q)}$ with a probability of at least $1-O(\min(n_p,n_q)^{-3})$. 
    \end{itemize}
\end{restatable}
To prove this, we employ an approach inspired by a technique for expanding bipartite graphs in the random assignment theory~\cite{FriezeJohansson2017,Frieze2021,Talagrand1995}. Although Frieze and Johansson~\cite{FriezeJohansson2017} and Frieze~\cite{Frieze2021} consider the assignment problem with minimum cost, we use a similar argument to analyze maximum-weight matchings. 
Specifically, we analyze edge weights of a maximum-weight matching in a bipartite graph by considering an alternating cycle of a maximum-weight matching in the bipartite graph restricted to the ``heavy'' edges and bounding its length. Furthermore, to demonstrate that such an alternating cycle can be found, we show ``expander'' properties of sub-bipartite graphs of certain size.

In the proof of Lemma~\ref{lem:no_heavy_edge}, we use the Chernoff bound, which can be used to bound the tails of the distribution for sums of independent random variables.
\begin{lemma}[Chernoff bound]
Let $X_1,X_2,\ldots,X_d$ be independent random variables such that $0\leq X_i \leq 1$ for all $i\in [d]$. Let $S = \sum_{i=1}^d X_i$. Then, for all $\varepsilon > 0$, we have
\begin{align*}
    &\mathrm{(i)\ } \mathbf{Pr}[S \geq (1+\varepsilon) \mathbb{E}[S]] \leq \exp \left(- \frac{\varepsilon^2 \mathbb{E}[S]}{3} \right), \\
    &\mathrm{(ii)\ } \mathbf{Pr}[S \leq (1-\varepsilon) \mathbb{E}[S]] \leq \exp \left(- \frac{\varepsilon^2 \mathbb{E}[S]}{2} \right). 
\end{align*}
\end{lemma}

Let $H$ be a complete bipartite graph with bipartition $(L,R)$ with weights taking a value in $[0,1]$.
Let $w_0$ be a constant on $[0,1]$, and $w(i,j)$ be the weight on the edge $\{i,j\}$ for each $i\in L$ and $j\in R$.
Fix $r \in \{1,2,\ldots,\min(|L|,|R|)\}$.
Let $M_{\mathrm{OPT}}$ denote the maximum-weight matching of size $r$ in $H$. 

An alternating cycle $C$ of $M_{\mathrm{OPT}}$ such that 
$$
C = (i_1, M_{\mathrm{OPT}}(i_2), i_2,\ldots,M_{\mathrm{OPT}}(i_{\ell}),i_{\ell}, M_{\mathrm{OPT}}(i_1),i_1)
$$
is \emph{acceptable} for $w_0$ if $w(i_t, M_{\mathrm{OPT}}(i_{t+1})) \geq w_0$ for every $t = 1,2,\ldots,\ell-1$ and $w(i_{\ell}, M_{\mathrm{OPT}}(i_1)) \geq w_0$.
Similarly, an alternating path $P$ of $M_{\mathrm{OPT}}$ such that 
$$
P = (M_{\mathrm{OPT}}(i_1), i_1, M_{\mathrm{OPT}}(i_2), i_2,\ldots,M_{\mathrm{OPT}}(i_{\ell}),i_{\ell}, j')
$$ 
is \emph{acceptable}  for $w_0$ if $j'\in R$ is not incident to any edges in $M_{\mathrm{OPT}}$, $w(i_t, M_{\mathrm{OPT}}(i_{t+1})) \geq w_0$ for every $t = 1,2,\ldots,\ell-1$ and $w(i_{\ell}, j') \geq w_0$.
We refer the number of the left vertices that appear in an alternating cycle (resp. path) as the \emph{length} of the alternating cycle (resp. path).
\begin{lemma}\label{lemma:the_weight_of_at_least}
If there exists an acceptable alternating path of $M_{\mathrm{OPT}}$ of length $\ell$ for $w_0$,
or an acceptable alternating cycle of $M_{\mathrm{OPT}}$ of length $\ell$ for $w_0$,
then every edge in matching $M_{\mathrm{OPT}}$ has a weight at least $1-\ell\cdot (1-w_0)$.
\end{lemma}
\begin{proof}
    Consider the case where there is an acceptable alternating path of $M_{\mathrm{OPT}}$ of length $\ell$ for $w_0$. 
    Let $P = (M_{\mathrm{OPT}}(i_1), i_1, M_{\mathrm{OPT}}(i_2), i_2,\ldots,$ $M_{\mathrm{OPT}}(i_{\ell}),i_{\ell}, j')$ denote the acceptable alternating path.
    We construct a new matching $M_{\mathrm{new}}$ from $M_{\mathrm{OPT}}$ by replacing the edges in $P \cap M_{\mathrm{OPT}}$, namely,
    $$
    M_{\mathrm{new}} = (M_{\mathrm{OPT}} \setminus \{\{i_1,M_{\mathrm{OPT}}(i_{1})\}, \ldots,\{i_{\ell},M_{\mathrm{OPT}}(i_{\ell})\}\}) \cup \{\{i_1,M_{\mathrm{OPT}}(i_{2})\}, \ldots,\{i_{\ell},j'\}\}.
    $$
    We have $M_{\mathrm{new}}(i') = M_{\mathrm{OPT}}(i')$ for every $i'\in L$ with $i' \notin \{i_1,i_2,\ldots,i_{\ell}\}$ that appears in matching $M_{\mathrm{new}}$.
    Since $M_{\mathrm{OPT}}$ is the maximum-weight matching of size $r$, we obtain $\sum_{t=1}^{\ell} w(i_t, M_{\mathrm{OPT}}(i_t)) \geq \sum_{t=1}^{\ell} w(i_t, M_{\mathrm{new}}(i_t))$.
    Since $w(i,j)\leq 1$ for all $i\in L$ and $j\in R$, we have that
    for all $\{i, M_{\mathrm{OPT}}(i)\} \in M_{\mathrm{OPT}}$,
    \begin{align*}
    \sum_{t=1}^{\ell} w(i_t, M_{\mathrm{OPT}}(i_t)) 
    &= 
        w(i, M_{\mathrm{OPT}}(i)) + \sum_{i'\in \{i_1,i_2,\ldots,i_{\ell}\} \setminus \{i\}} w(i', M_{\mathrm{OPT}}(i')) \\
    &\le 
        w(i, M_{\mathrm{OPT}}(i)) + \ell -1.
    \end{align*}
    Thus, for all $\{i, M_{\mathrm{OPT}}(i)\} \in M_{\mathrm{OPT}}$, we obtain
    \begin{align*}
        w(i, M_{\mathrm{OPT}}(i)) 
        &\geq 
            \sum_{t=1}^{\ell} w(i_t, M_{\mathrm{OPT}}(i_t)) - \ell + 1 \\
        &\geq 
            \sum_{t=1}^{\ell} w(i_t, M_{\mathrm{new}}(i_t)) - \ell + 1 \\
        &= 
            \sum_{t=1}^{\ell-1} w(i_t, M_{\mathrm{OPT}}(i_{t+1})) + w(i_{\ell}, M_{\mathrm{OPT}}(i_{1})) - \ell + 1 \\
        &\geq 
            \ell\cdot w_0 - \ell + 1 \\
        &=  
            1 - \ell\cdot (1-w_0).
    \end{align*}
    This completes the proof.
    When there exists an alternating cycle of $M_{\mathrm{OPT}}$ of length $\ell$ for $w_0$, we can prove the statement similarly.
\end{proof}



\subsubsection{Proof of $($a$)$ of Lemma~\ref{lem:no_heavy_edge}}

First, consider $($a$)$ of Lemma~\ref{lem:no_heavy_edge}. Lemma~\ref{lemma:the_weight_of_at_least} and the following lemma (Lemma~\ref{lemma:the_length_of_alternating}) immediately yield $($a$)$ of Lemma~\ref{lem:no_heavy_edge}.
%
\begin{lemma}\label{lemma:the_length_of_alternating}
Suppose that distribution $\mathcal{D}$ is $(\alpha,\beta)$-PDF-bounded. Let $c_p$ be a sufficiently large constant such that $c_p > 1800\alpha^{-1} > 0$. Let $w_0 = 1-c_p \frac{\log n_p}{n_p}$.
Let $\bar{M}$ denote the maximum-weight matching of a fixed size $r$ between $N_p$ and $M(N_p)$.
\begin{itemize}
    \item[$(\mathrm{i})$]  For each $r=1,2,\ldots,\lfloor\frac{n_p}{300}\rfloor$, the probability that there exists an acceptable alternating path of $\bar{M}$ for $w_0$ of length $O(\log n_p)$ for $w_0$ is at least $1 - O(n_p^{-3})$. 
    \item[$(\mathrm{ii})$]  Furthermore, for each $r=\lfloor\frac{n_p}{300}\rfloor+1,\lfloor\frac{n_p}{300}\rfloor+2,\ldots,n_p$ the probability that there exists an acceptable alternating cycle of $\bar{M}$ for $w_0$ of length $O(\log n_p)$ is at least $1 - O(n_p^{-3})$. 
\end{itemize}
\end{lemma}

\begin{proof}
Let $W(i,j)$ denote the random variable which represents the weight of the edge $\{i,j\}$ in $G$. 
For each subset $S\subseteq N_p$, we define $\mathrm{\Gamma}(S)= \{ j \in M(N_p) \mid \{i,j\}\in E \land W(i,j)\geq w_0~\text{for some}\  i \in S \}$.

We first fix any size $r \in [n_p]$. Let $\bar{M}$ denote the maximum-weight matching of the fixed size $r$ between $N_p$ and $M(N_p)$.
Let $\bar{M}(N_p) \subseteq M(N_p)$ represent the 
We then fix an agent $i\in N_p$, which is incident to an edge in matching $\bar{M}$.
We define a sequence of sets $(S_t)_{t\geq 0}$ as follows. First, we set $S_0=\{i\}$.
If $\mathrm{\Gamma}(S_0) \subseteq \bar{M}(N_p)$, then let $S_1$ be the set of vertices in $N_p$ which are matched to some vertex in $\mathrm{\Gamma}(S_0)$ under matching $\bar{M}$. 
Similarly, for $t \geq 2$, as long as $\mathrm{\Gamma}(S_{t-1}) \subseteq \bar{M}(N_p)$ holds, let $S_t \subseteq N_p$ be the set of vertices which are matched to some vertex in $\mathrm{\Gamma}(S_{t-1})$ under matching $\bar{M}$.
If, at some step $\ell$, $\mathrm{\Gamma}(S_{\ell}) \setminus \bar{M}(N_p) \neq \emptyset$ occurs, then there exists an acceptable alternating path $P$ of $\bar{M}$ such that $P = (i_1 = i, \bar{M}(i_2), i_2,\ldots, \bar{M}(i_{\ell}),i_{\ell}, j')$ where $j' \in M(N_p) \setminus \bar{M}(N_p)$, and $W(i_t, \bar{M}(i_{t+1})) \geq w_0$ for every $t = 1,2,\ldots,\ell-1$ and $W(i_{\ell}, j') \geq w_0$.
Moreover, if $i\in S_{\ell}$ holds at some  step $\ell$, then we find  an acceptable alternating cycle of $\bar{M}$.

\paragraph{Proof of (i).}
We will show that for each $r=1,2,\ldots,\lfloor\frac{n_p}{300}\rfloor$, the number of steps $\ell$ for which $\mathrm{\Gamma}(S_{\ell}) \setminus \bar{M}(N_p) \neq \emptyset$ occurs is at most $O(\log n_p)$ with a probability of at least $1-O(n_p^{-3})$,
and that for each $r=\lfloor\frac{n_p}{300}\rfloor + 1, \lfloor\frac{n_p}{300}\rfloor + 2,\ldots, n_p$, the number of steps $\ell$ for which $i \in S_{\ell}$ happens is at most $O(\log n_p)$ with a probability of at least $1-O(n_p^{-3})$.

To this end, we select a size-$n_p$ subset of $I$ uniformly at random, denoted by $R_u \subseteq I$. 
For each subset $S\subseteq N_p$, let $\mathrm{\Gamma}_u(S)= \{ j \in R_u \mid \{i,j\}\in E \land W(i,j)\geq w_0  ~ \text{for some}\  i \in S \}$. 
Let ${\rho}_{\mathcal{D}} = \mathbf{Pr}[X \geq w_0 ]$, where $X$ is a random variable drawn from the probability distribution $\mathcal{D}$. From PDF-boundedness of $\mathcal{D}$, we have ${\rho}_{\mathcal{D}} \geq \alpha(1 - w_0) \geq \alpha c_p \log n_p/n_p$.
Let $\text{Bin}(n,p)$ denote the binomial distribution with parameters $n$ and $p$. 
For a fixed $S \subseteq N_p$, $|\mathrm{\Gamma}_u(S)|$ is distributed as $\mathrm{Bin}(n_p, {\rho}_{S})$,
where ${\rho}_{S} = 1-(1-{\rho}_{\mathcal{D}})^{|S|}$,
and we have $\mathbb{E}[|\mathrm{\Gamma}_u(S)|] = n_p {\rho}_{S}$.

Firstly, we consider the case when $r \in \{1,2,\ldots,\lfloor\frac{1}{2\alpha(1-w_0)} \rfloor\}$.
Let $S\subseteq N_p$ be a subset such that $ 1 \leq |S| \leq \frac{1}{2\alpha(1-w_0)}$. We now have 
$$
\rho_S \geq  1-\left(1-\alpha(1-w_0)\right)^{|S|}
\geq 1 - \mathrm{e}^{-\alpha(1-w_0) |S|} \geq \frac{\alpha (1-w_0)|S|}{2}.
$$
Here, we utilize two inequalities: $(1-t)^x \leq \mathrm{e}^{-t x}$ for $x\geq1$ and $0\leq t \leq 1$, and $1-\mathrm{e}^{-x} \geq \frac{x}{2}$ for $0\leq x\leq\frac{1}{2}$.
Thus, we obtain $\mathbb{E}[|\mathrm{\Gamma}_u(S)|] = n_p {\rho}_{S} \geq \frac{\alpha c_p |S| \log n_p}{2}$.

We now employ the Bayes' theorem technique, similar to the one used to prove Inequality~\eqref{ineq:condioning_equation_3}. 
We claim that, for any given $S\subseteq N_p$ such that $ 1 \leq |S| \leq \frac{1}{2\alpha(1-w_0)}$,
\begin{equation}\label{inequality:Bayse_2}
    \mathbf{Pr}\left[ |\mathrm{\Gamma}_u(S)|\leq \frac{\alpha c_p|S|\log n_p}{4}\, \middle\vert\, R_u = M(N_p) \right]
    \leq \mathbf{Pr}\left[ |\mathrm{\Gamma}_u(S)|\leq \frac{\alpha c_p|S|\log n_p}{4} \right].
\end{equation}
The subset $\mathrm{\Gamma}_u(S)$ represents the set of items adjacent to an agent in $S$ by edges with weights above a certain threshold, which implies that these items have a certain minimum value for an agent in $S$. The condition $|\mathrm{\Gamma}_u(S)|\leq \alpha c_p|S|\log n_p/4$ does not increase
the probability of these items in $R_u$ being chosen by an agent in $N_p$ in the round-robin algorithm. 
Therefore, if $S \subseteq N_p$ is fixed and $|\mathrm{\Gamma}_u(S)|\leq \alpha c_p|S|\log n_p/4$ is conditioned, then the probability of $R_u=M(N_p)$ does not increase compared to the case without the condition. 
Consequently, $\mathbf{Pr}\big[R_u = M(N_p) \mid  |\mathrm{\Gamma}_u(S)|\leq \alpha c_p|S|\log n_p/4 \big] \leq \mathbf{Pr}\big[ R_u = M(N_p)  \big]$. Hence, by applying Bayes' theorem, we obtain Inequality~\eqref{inequality:Bayse_2}.

Let $I[n_p]$ denote the family of size-$n_p$ subsets of $I$. We now have, for all $S\subseteq N_p$ such that $ 1 \leq |S| \leq \frac{1}{2\alpha(1-w_0)}$,
\begin{align*}
    &\mathbf{Pr}\left[ \exists\, S \subseteq N_p \text{ such that } 1 \leq |S| \leq \frac{1}{2\alpha(1-w_0)} \text{ and } |\mathrm{\Gamma}(S)|\leq \frac{\alpha c_p|S|\log n_p}{4} \right] \\
    &\leq
        \sum_{R_u \in I[n_p]}\  \sum_{S \subseteq N_P \text{ with } 1 \leq |S| \leq \frac{1}{2\alpha(1-w_0)} }  \mathbf{Pr}\left[ |\mathrm{\Gamma}_u(S)|\leq \frac{\alpha c_p|S|\log n_p}{4}\, \middle\vert \, R_u = M(N_p) \right] \\
    &\quad\cdot 
        \mathbf{Pr}[R_u = M(N_p)] \\
    &\leq 
        \sum_{S \subseteq N_P \text{ with } 1 \leq |S| \leq \frac{1}{2\alpha(1-w_0)} } \mathbf{Pr}\left[ |\mathrm{\Gamma}_u(S)|\leq \frac{\alpha c_p|S|\log n_p}{4} \right]\\
    &\leq
        \sum_{S \subseteq N_P \text{ with } 1 \leq |S| \leq \frac{1}{2\alpha(1-w_0)} } \mathbf{Pr}\Big[ |\mathrm{\Gamma}_u(S)|\leq \frac{1}{2} \cdot \mathbb{E}\big[|\mathrm{\Gamma}_u(S)|\big] \Big].
\end{align*}
Since $|\mathrm{\Gamma}_u(S)|$ is a sum of $n_p$ independent random variables, we can apply the Chernoff bound with $d = n_p$ and $\varepsilon = \frac{1}{2}$. This yields 
$
    \mathbf{Pr}\left[ |\mathrm{\Gamma}_u(S)|\leq \frac{1}{2} \cdot \mathbb{E}[|\mathrm{\Gamma}_u(S)|] \right]
    \leq \mathrm{exp}\left(-\frac{\mathbb{E}[|\mathrm{\Gamma}_u(S)|]}{8}\right)
    \leq \mathrm{exp}\left(-\frac{\alpha c_p |S|\log n_p}{16}\right).
$
Finally, the assumption of $\alpha c_p \geq 160$ leads 
\begin{align*}
    &\mathbf{Pr}\left[ \exists\, S \subseteq N_p \text{ such that } 1 \leq |S| \leq \frac{1}{2\alpha(1-w_0)} \text{ and } |\mathrm{\Gamma}(S)|\leq \frac{\alpha c_p|S|\log n_p}{4} \right] \\
    &\leq \sum_{s=1}^{ \frac{1}{2\alpha(1-w_0)}} \binom{n_p}{s} \cdot\mathrm{exp}\left(-\frac{\alpha c_p s\log n_p}{16}\right) \\
    &\leq
    \sum_{s=1}^{\frac{1}{2\alpha(1-w_0)}} n_p^{\left(1-\frac{\alpha c_p}{16}\right)s}  \\
    &\leq
    \sum_{s=1}^{ \frac{1}{2\alpha(1-w_0)}} n_p^{-9s}  
    =  O(n_p^{-3}).
\end{align*}
Therefore, for any subset $S \subseteq N_p$ such that $ 1 \leq |S| \leq \frac{1}{2\alpha(1-w_0)}$, we have $|\mathrm{\Gamma}(S)| \ge \frac{\alpha c_p|S|\log n_p}{4}$ with a probability of at least $1-O(n_p^{-3})$.
Hence, in the sequence of sets $(S_t)_{t\geq 0}$, if $1 \leq |S_t| \leq \frac{1}{2\alpha(1-w_0)}$, then $\frac{\alpha c_p \log n_p}{4} |S_t| \leq |\mathrm{\Gamma}(S_t)| = |S_{t+1}|$ with a probability of at least $1-O(n_p^{-3})$. 
This means that the size of $S_t$ increases by a factor of $\frac{\alpha c_p \log n_p}{4}$ until it exceeds $r$.
This implies that, for the sequence of sets $(S_t)_{t\geq 0}$, there exists $t_1 = \log n_p$ such that $|\mathrm{\Gamma}(S_{t_1})| > r$
with a probability of at least $1-O(n_p^{-3})$ since $r \leq \frac{1}{2\alpha(1-w_0)} = \frac{n_p}{2 \alpha c_p \log n_p} < \left(\frac{\alpha c_p \log n_p}{4}\right)^{t_1}$.
Thus, we can find an acceptable alternating path of length $O(\log n_p)$ with a probability of at least $1-O(n_p^{-3})$.

Secondly, we consider the case where $r \in  \{\lfloor \frac{1}{2\alpha(1-w_0)} \rfloor + 1,  \lfloor \frac{1}{2\alpha(1-w_0)} \rfloor+2,\ldots, \lfloor \frac{n_p}{300} \rfloor \}$.
We examine a subset $S\subseteq N_p$ such that $\frac{1}{2\alpha(1-w_0)} < |S| \leq \frac{n_p}{300}$.
In this case, we have
$$
\mathbb{E}\big[|\mathrm{\Gamma}_u(S)|\big] = n_p \rho_S
\geq n_p \left(1-\left(1-\alpha(1-w_0)\right)^{|S|}\right) 
\geq n_p \left(1 - \mathrm{e}^{-\alpha(1-w_0) |S|} \right) > n_p\left(1- \mathrm{e}^{-\frac{1}{2}}\right)
\geq \frac{n_p}{4}.
$$
Here, we utilize $(1-t)^x \leq \mathrm{e}^{-t x}$ for $x\geq1$ and $0\leq t \leq 1$, and $1- \mathrm{e}^{-\frac{1}{2}} \geq \frac{1}{4}$.

By a similar discussion to that which we use to obtain Inequality~\eqref{inequality:Bayse_2}, we obtain
\begin{align*}
\mathbf{Pr}\left[ \exists\, S\subseteq N_p \text{ such that } |\mathrm{\Gamma}(S)|\leq \frac{n_p}{8} \right] 
&\leq 
    \mathbf{Pr}\left[ \exists\, S\subseteq N_p \text{ such that } |\mathrm{\Gamma}_u(S)|\leq \frac{n_p}{8} \right] \\
&\leq  
    \sum_{S \subseteq N_P } \mathbf{Pr}\left[ |\mathrm{\Gamma}_u(S)|\leq \frac{|S|}{8} \right]\\
&\leq
    \sum_{S \subseteq N_P } \mathbf{Pr}\Big[ |\mathrm{\Gamma}_u(S)|\leq \frac{1}{2} \cdot \mathbb{E}\big[|\mathrm{\Gamma}_u(S)|\big] \Big].
\end{align*}
By the Chernoff bound with $d=n_p$ and $\varepsilon = \frac{1}{2}$, we get
$\mathbf{Pr}\Big[ |\mathrm{\Gamma}_u(S)|\leq \frac{1}{2} \cdot \mathbb{E}\big[|\mathrm{\Gamma}_u(S)|\big] \Big] \leq \mathrm{exp}\left(-\frac{n_p}{32}\right)$.
Thus,
\begin{align*}
    \mathbf{Pr}\left[ \exists\, S\subseteq N_p \text{ such that } |\mathrm{\Gamma}(S)|\leq \frac{n_p}{8} \right]
    &\leq 
        \sum_{s=\frac{n_p}{2 \alpha c_p \log n_p}}^{\frac{n_p}{300}} \binom{n_p}{s} \,\mathrm{exp}\left(-\frac{n_p}{32}\right) \\
    &\leq 
        \sum_{s=\frac{n_p}{2 \alpha c_p \log n_p}}^{\frac{n_p}{300}} \left(\frac{\mathrm{e}n_p}{s}\right)^s \,\mathrm{exp}\left(-\frac{n_p}{32}\right) \\
     &\leq 
        n_p (300\mathrm{e})^{\frac{n_p}{300}} \mathrm{e}^{-\frac{n_p}{32}} \\
     &=
        n_p \mathrm{e}^{-\left(\frac{1}{32} - \frac{1+\log 300}{300} \right)n_p} 
     = O(n_p^{-3}).
\end{align*}

Therefore, when $\frac{1}{2\alpha(1-w_0)} < r \leq \frac{n_p}{300}$, there exists $t_1+1 = O(\log n_p)$ such that $|\mathrm{\Gamma}(S_{t_1})| > \frac{n_p}{8} > \frac{n_p}{300} > r$
with a probability of at least $1-O(n_p^{-3})$.
This implies that there exists an acceptable alternating path of length $O(\log n_p)$ of $M_p^r$.

\paragraph{Proof of (ii).}
Next, we consider each round $r$ such that $r = \lfloor \frac{n_p}{300} \rfloor+1,\lfloor \frac{n_p}{300} \rfloor+2,\ldots, n_p$.
In this case, we will prove that there exists an acceptable alternating cycle of length $O(\log n_p)$ of $\bar{M}$ instead of an acceptable alternating path. 
For a subset $S\subseteq N_p$ such that $\frac{n_p}{300} < |S| $, 
we have 
$$
1-\rho_S = (1-\rho_{\mathcal{D}})^{|S|} 
< (1-\rho_{\mathcal{D}})^{\frac{n_p}{300}} 
\leq \left(1- \frac{\alpha c_p \log n_p}{n_p}\right)^{\frac{n_p}{300}} 
\leq \mathrm{e}^{-\frac{\alpha c_p}{300}\log n_p}
\leq n_p^{-\frac{\alpha c_p}{300}}.
$$ 
We consider the probability that there exists a subset $S \subseteq N_p$ such that $\frac{n_p}{300} < |S| $ and $n_p - |\mathrm{\Gamma}(S)| \leq \frac{1}{2} \left( n_p - |S| \right)$.
By a Bayes's theorem technique which we utilize to get Inequality~\eqref{inequality:Bayse_2}, we have
\begin{align*}
    \mathbf{Pr}\left[\exists\, S \subseteq N_p \text{ such that }  n_p - |\mathrm{\Gamma}(S)| > \frac{1}{2} \left( n_p - |S| \right) \right] 
    \leq
    \mathbf{Pr}\left[\exists\, S \text{ such that }  n_p - |\mathrm{\Gamma}_u(S)| > \frac{1}{2} \left( n_p - |S| \right) \right].
\end{align*}
Then,
\begin{align*}
    &\mathbf{Pr}\left[\exists\, S \subseteq N_p \text{ such that }  n_p - |\mathrm{\Gamma}(S)| > \frac{1}{2} \left( n_p - |S| \right) \right] \\
    &\le
    \mathbf{Pr}\left[\exists\, S \text{ such that }  n_p - |\mathrm{\Gamma}_u(S)| > \frac{1}{2} \left( n_p - |S| \right) \right] \\
    &\leq \sum_{t=0}^{\frac{ 299 n_p}{300}}
    \mathbf{Pr}\left[|S|=n_p-t, |\mathrm{\Gamma}_{u}(S)|\leq n_p - t/2 \right]  \\
    &= \sum_{t=0}^{\frac{ 299 n_p}{300}}
    \mathbf{Pr}\left[\exists S\subseteq N_p, T\subseteq R_u : |S|=n_p-t,\,|T|=n_p - t/2, \mathrm{\Gamma}_{u}(S)\subseteq T \right]  \\
    &\leq \sum_{t=0}^{\frac{ 299 n_p}{300}}
    \sum_{\substack{\exists S\subseteq N_p, T\subseteq R_u \\ |S|=n_p-t,\, |T|=n_p-t/2}} \mathbf{Pr}\left[\mathrm{\Gamma}_{u}(S)\subseteq T  \right]  \\
    &\leq \sum_{t=0}^{\frac{299 n_p}{300}} \binom{n_p}{t}\cdot \binom{n_p}{t/2} \cdot (1-\rho_{S})^{t/2}  \\
    &\leq \sum_{t=1}^{\frac{ 299 n_p}{300}} 
    \left(\frac{\mathrm{e}n_p}{t}\right)^{t} \cdot
    \left(\frac{2\mathrm{e}n_p}{t}\right)^{t/2} \cdot
    n_p^{-\frac{\alpha c_p t}{600}} \\
    &\leq \sum_{t=1}^{\frac{ 299 n_p}{300}} 
    \left(\frac{2\mathrm{e}n_p}{t}\right)^{2t} \cdot
    n_p^{-\frac{\alpha c_p t}{600}} \\
    &\leq \sum_{t=1}^{\frac{ 299 n_p}{300}} 
    \left(\frac{2\mathrm{e}n_p}{t}\right)^{2t} \cdot
    n_p^{-3t} \\
    &= O(n_p^{-3})
\end{align*}
where we utilize the fact that $\alpha c_p>1800$.
Therefore, for a subset $S\subseteq N_p$ with $\frac{n_p}{300} < |S| $, $n_p - |\mathrm{\Gamma}(S)| \leq \frac{1}{2} \left( n_p - |S| \right)$ holds with a probability of at least $1-O(n_p^{-3})$. Thus, in the sequence, if $\frac{n_p}{300} < |S_t| $, then
\[
    n_p - |S_{t+1}| \leq \frac{1}{2}(n_p - |S_t|) < \frac{299 n_p}{600},
\]
which implies $|S_{t+1}| > \frac{311 n_p}{600} > \frac{n_p}{300}$. 
Let $t_2=\left\lceil \log \frac{598 n_p}{600} \right\rceil $. 
Applying the above inequality repeatedly for $t=t_1+1,t_1+2,\ldots,t_1+t_2+1$, we get 
\[
    n_p - |S_{t_1+t_2}| \leq \frac{1}{2^{t_2}}(n_p - |S_{t_1}|) < \frac{600 }{598 n_p} \frac{299 n_p}{600} < 1.
\]
This means that $|S_{t_1+t_2+1}|=n_p$ and $S_{t_1+t_2+1}$ includes all the vertices in $N_p$. Thus, $i\in S_{t_1+t_2+1}$ and, therefore, we can find an acceptable alternating cycle of length $\ell = t_1 +  t_2 +1 \le 5\log n_p$.
\end{proof}


\subsubsection{Proof of $($b$)$ of Lemma~\ref{lem:no_heavy_edge}}

To prove $($b$)$ of Lemma~\ref{lem:no_heavy_edge}, we prove the following lemma. We will obtain $($b$)$ of Lemma~\ref{lem:no_heavy_edge} from Lemmas~\ref{lemma:the_weight_of_at_least} and~\ref{lemma:the_length_of_alternating_for_M_N_q}
\begin{lemma}\label{lemma:the_length_of_alternating_for_M_N_q}
Suppose that distribution $\mathcal{D}$ is $(\alpha,\beta)$-PDF-bounded. Let $c_q  >12 c_p + 80\alpha^{-1} $ be a constant, and $w_0 = 1-c_q \frac{(\log \min(n_p,n_q))^3}{\min(n_p,n_q)}$.
Let $\bar{M}$ denote the maximum-weight matching of a fixed size $r$ between $N_p$ and $M(N_q)$.
\begin{itemize}
    \item[$(\mathrm{i})$]  For each $r=1,2,\ldots,\lfloor\frac{\min(n_p,n_q)}{10}\rfloor$, between $N_p$ and $M^r(N_q)$, the probability that there exists an acceptable alternating path of $\bar{M}$ for $w_0$ of length $O(\log \min(n_p,n_q))$ is at least $1 - O(\min(n_p,n_q)^{-3})$, and
    \item[$(\mathrm{ii})$]  for each $r=\lfloor\frac{\min(n_p,n_q)}{10}\rfloor + 1 ,\lfloor\frac{\min(n_p,n_q)}{10}\rfloor + 2, \ldots, \min(n_p,n_q)$,  the probability that there exists an acceptable alternating cycle of $\bar{M}$ for $w_0$ of length $O(\log \min(n_p,n_q))$ is at least $1 - O(\min(n_p,n_q)^{-3})$. 
\end{itemize}
\end{lemma}

\begin{proof}
Let $w_0 = 1-c_q \frac{(\log \min(n_p,n_q))^3}{\min(n_p,n_q)}$. 
In contrast to the proof of Lemma~\ref{lemma:the_length_of_alternating}, we consider edges and matchings between $N_p$ and $M(N_q)$.
Let $W(i,j)$ denote the random variable which describes the weight of an edge $\{i,j\}$ in $G$.
First, we claim that for each $i\in N_p$ and $j\in M(N_q)$, 
there exists a constant $c_0$ such that $0 < c_0\le \alpha(c_q - 12c_p) $ and 
\begin{equation}\label{equation:W(i,j)_w_0}
    \mathbf{Pr}[W(i,j)\le w_0] 
     \le
    1-c_0\frac{(\log \min(n_p,n_q))^3}{\min(n_p,n_q)}.
\end{equation}
Let $i^r \in N_p$ be the agent that selects a new item in round $r$, and $j^r_p \in M(N_p)$ be the item such that $j^r_p \in M^r(N_p) \setminus M^{r-1}(N_p)$. Moreover, let $j^r_q \in M(N_q)$ be the item such that $j^r_q \in M^r(N_q) \setminus M^{r-1}(N_q)$.
Furthermore, let $P^r = (i^r, M^{r-1}(i_2), i_2, M^{r-1}(i_3),i_3,\ldots, M^{r-1}(i_{a}),i_{a}, j^r_p)$ be the alternating path that is found when class $p$ selects a new item in round $r$. Let $W(P^r) = W(i^r, M^{r-1}(i_2)) - W(M^{r-1}(i_2), i_2) + W(i_2, M^{r-1}(i_3)) - \cdots - W(M^{r-1}(i_{a}),i_{a}) + W(i_{a}, j^r_p)$. Here, we have $\{i^r, M^{r-1}(i_2)\}, \{i_2, M^{r-1}(i_3)\}, \ldots, \{i_{a}, j^r_p\} \in M^r$. Let $\ell(P^r)$ be the length of $P^r$.

We first consider the case where $p < q$. 
In this case, we have $W(i,j^r_q) \le W(P^r)$ for each $i = i^r,i^r+1,\ldots, i^{n_p}$.
From $($a$)$ of Lemma~\ref{lem:no_heavy_edge}, we have, at least a probability of $1-O(n_p^{-3})$, 
\[
    W(P^r) \ge (\ell(P_r)+1)\cdot \left(1-c_p\frac{(\log n_p)^2}{n_p}\right) - \ell(P_r) = 1 - (\ell(P_r)+1)\cdot c_p\frac{(\log n_p)^2}{n_p}
    \ge 1 - 6c_p\frac{(\log n_p)^3}{n_p}.
\]
Let $d\in [0,1]$ be a real number and $\mathcal{D}_{\le d}$ be  the conditional distribution of $\mathcal{D}$ on $[0,d]$.
Let$X \sim \mathcal{D}$, $X_1 \sim \mathcal{D}_{\le W(P_r)}$ and $X_2 \sim \mathcal{D}_{\le 1 - 6c_p(\log n_p)^2/n_p}$. Then, for each $i = i^r,i^r+1,\ldots, i^{n_p}$, we have
\begin{align}\label{inequality:c_q_6c_p}
    \mathbf{Pr}[W(i,j^r_q)\ge w_0] 
    &= \mathbf{Pr}[X_1\ge w_0] \nonumber\\
    &\ge \mathbf{Pr}[X_2\ge w_0] \nonumber\\
    &= \mathbf{Pr}[X \ge w_0 \mid X \le 1 - 6c_p (\log n_p)^3/n_p] \nonumber\\
    &\ge \mathbf{Pr}[w_0 \le X \le 1 - 6c_p (\log n_p)^3/n_p] \nonumber\\
    &\ge \alpha \cdot \mathbb{E}[1 - 6c_p (\log n_p)^3/n_p - w_0] \nonumber\\
    &\ge \alpha \cdot \left( 1 - 6c_p \frac{(\log n_p)^3}{n_p} - 1+c_q \frac{(\log \min(n_p,n_q))^3}{\min(n_p,n_q)}\right)\cdot \left( 1 - O(n_p^{-3})\right) \nonumber\\
    &\ge \alpha \cdot (c_q - 6c_p) \frac{(\log \min(n_p,n_q))^3}{\min(n_p,n_q)} + O(n_p^{-3}).
\end{align}
Here, we utilize that $\frac{(\log \min(n_p,n_q))^3}{\min(n_p,n_q)} \ge \frac{(\log n_p)^3}{n_p}$. Since $c_q - 6c_p > 0$, we have $\mathbf{Pr}[W(i,j^r_q)\le w_0] =
    1-O\left(\frac{(\log \min(n_p,n_q))^3}{\min(n_p,n_q)}\right)$ for all $i = i^r,i^r+1,\ldots, i^{n_p}$.

Secondly, we consider agents $i = i^1,i^2,\ldots, i^{r-1}$, and conditions on $W(i, j^r_q)$. When class $p$ chooses item $j^r_p$ in round $r$, a condition is added to the weight on edge $\{i,j^r_q\}$. 
That is, the weight of any alternating path starting from $i^r$, passing through $i$, and reaching $j^r_q$ is bounded above by $W(P^r)$. Let $P_i$ denote such a path from $i^r$ to $i$, i.e., $P_i = (i^r,i_2,M^{r-1}(i_2),\ldots,M^{r-1}(i),i)$. For path $P_i$, we have $W(P_i) = W(i^r,i_2) - W(i_2,M^{r-1}(i_2)) + \cdots - W(M^{r-1}(i),i) \ge \ell(P_i) - \ell(P_i) \cdot \left(1-5c_p \frac{(\log n_p)^2}{n_p}\right) \ge  6c_p \frac{(\log n_p)^3}{n_p}$ with a probability of at least $1-O(n_p^{-3})$. Hence, $W(P^r) - \max W(P_i) \ge  1 - 12c_p\frac{(\log n_p)^3}{n_p}$ with a probability of at least $1-O(n_p^{-3})$.
Then, from a discussion similar to Inequality~\eqref{inequality:c_q_6c_p}, we obtain 
\[
\mathbf{Pr}[W(i,j^r_q)\ge w_0] \ge 
    \alpha \cdot (c_q - 12c_p) \frac{(\log \min(n_p,n_q))^3}{\min(n_p,n_q)} + O(n_p^{-3}).
\]
for all $i = i^1,i^2,\ldots, i^{r-1}$. Since $c_q > 12 c_p$, we conclude that $\mathbf{Pr}[W(i,j)\le w_0] \le 1- c_0\left(\frac{(\log \min(n_p,n_q))^3}{\min(n_p,n_q)}\right)$, where $0 < c_0\le \alpha(c_q - 12 c_p)$.

When $p>q$, we can show Inequality~\eqref{equation:W(i,j)_w_0} similarly.

We now prove that for each $r=1,2,\ldots,\lfloor\frac{\min(n_p,n_q)}{300}\rfloor$, between $N_p$ and $M^r(N_q)$, the probability that there exists an acceptable alternating path of length $O(\log \min(n_p,n_q))$ is at least $1 - O(\min(n_p,n_q)^{-3})$.
We fix size $r$ and let $\bar{M}$ be the maximum-weight matching of size $r$ between $N_p$ and $M(N_q)$ in $G$.
For each subset $S\subseteq N_p$, let $\mathrm{\Gamma}(S)= \{ j \in M(N_q) \mid \{i,j\}\in E \land W(i,j)\geq w_0~\text{for some}\  i \in S \}$.
Similarly to the discussion on sequences of sets in the proof of Lemma~\ref{lemma:the_length_of_alternating}, we fix $i\in N_p$.
We define $S_0=\{i\}$.
If $\mathrm{\Gamma}(S_0) \subseteq \bar{M}(N_p)$, then let $S_1$ be the set of vertices in $N_p$ which are matched to some vertex in $\mathrm{\Gamma}(S_0)$ under matching $\bar{M}$. 
For each $t = 2,3,\ldots$, if $\mathrm{\Gamma}(S_{t-1}) \subseteq \bar{M}(N_p)$ holds, then we define $S_t \subseteq N_p$ as the set of vertices which are matched to some vertex in $\mathrm{\Gamma}(S_{t-1})$ under matching $\bar{M}$.
If, at some step $\ell$, $\mathrm{\Gamma}(S_{\ell}) \setminus \bar{M}(N_p) \neq \emptyset$ occurs, then there exists an acceptable alternating path of $\bar{M}$.
If, at some  step $\ell$, $i\in S_{\ell}$ holds, then we find  an acceptable alternating cycle of $\bar{M}$.

First, for a subset $S\subseteq N_p$ such that $ 1\le |S| \le \frac{\min(n_p,n_q)}{20}$,
then we have
\begin{align*}
        &\mathbf{Pr}\left[\exists\, S \text{ such that } |\mathrm{\Gamma}(S)|  \leq 2|S| \right]  \\
        &\leq \sum_{s=1}^{\min(n_p,n_q)/10}
        \mathbf{Pr}\left[|S|=s, |\mathrm{\Gamma}(S)|\leq 2s-1 \right]  \\
        &= \sum_{s=1}^{\min(n_p,n_q)/10}
        \mathbf{Pr}\left[\exists S\subseteq N_p, T\subseteq M(N_q) : |S|=s,\,|T|=2s-1, \mathrm{\Gamma}(S)\subseteq T \right]  \\
        &\leq \sum_{s=1}^{\min(n_p,n_q)/10}
        \sum_{\substack{\exists S\subseteq N_p, T\subseteq M(N_q) \\ |S|=s,\, |T|=2s-1}} \mathbf{Pr}\left[\mathrm{\Gamma}(S)\subseteq T  \right]  \\
        &\leq \sum_{s=1}^{\min(n_p,n_q)/10}\ 
        \sum_{\substack{\exists S\subseteq N_p, T\subseteq M(N_q) \\ |S|=s,\, |T|=2s-1}}\  \prod_{i\in N_p,\, j\in M(N_q) \setminus T}\mathbf{Pr}\left[W(i,j) \le w_0 \right]  \\
        &\leq \sum_{s=1}^{\min(n_p,n_q)/10}
        \binom{n_p}{s}\cdot \binom{n_q}{2s-1}\cdot  \left( 1 - c_0 \frac{(\log \min(n_p,n_q))^3}{\min(n_p,n_q)} \right)^{s(n_q-2s+1)}  \\
        &\leq \sum_{s=1}^{\min(n_p,n_q)/10} n_p^s\cdot n_q^{2s}\cdot \mathrm{exp}\left(- c_0 \frac{(\log \min(n_p,n_q))^3}{\min(n_p,n_q)} s(n_q-2s+1)\right)  \\
        &= \sum_{s=1}^{\min(n_p,n_q)/10} \mathrm{exp}\left(- c_0 \frac{(\log \min(n_p,n_q))^3}{\min(n_p,n_q)} s(n_q-2s+1) +s \log n_p +2s\log n_q\right) \\
        &\le \sum_{s=1}^{\min(n_p,n_q)/10} \mathrm{exp}\left(- \frac{c_0}{20} s(\log \min(n_p,n_q))^3 +s \log n_p +2s\log n_q\right) \\
        &\le \sum_{s=1}^{\min(n_p,n_q)/10}  \mathrm{exp}\left(s\cdot \left(- \frac{c_0}{20} (\log \min(n_p,n_q))^3 + O((\log \min(n_p,n_q) )^2)\right) \right)\\
        &= O(\min(n_p,n_q)^{-3}).
    \end{align*}
    Here, we utilize $\frac{n_q-\min(n_p,n_q)/10-1}{\min(n_p,n_q)} \ge 1/20$ and $c_0 \ge 80$. We also use condition~\ref{item:assumption-c}, i.e., $\max(n_p,n_q) \leq C\cdot \min(n_p,n_q)^{5/4}$.
Hence, in the sequence of sets $(S_t)_{t\geq 0}$, there exists $t_1 = \log \min(n_p,n_q)$ such that $|\mathrm{\Gamma}(S_{t_1})| > r$ when $r<\min(n_p,n_q)/10$
with a probability of at least $1-O(\min(n_p,n_q)^{-3})$.
Thus, we can find an acceptable alternating path of length $O(\log \min(n_p,n_q))$ with a probability of at least $1-O(\min(n_p,n_q)^{-3})$.

Second, we consider a subset $S\subseteq N_p$ such that $\frac{\min(n_p,n_q)}{10} < |S| \leq \min(n_p,n_q)$. Let $t=\min(n_p,n_q)-|S|$. Then,
\begin{align*}
    &\mathbf{Pr}\left[\exists\, S \text{ such that }  \min(n_p,n_q) - |\mathrm{\Gamma}(S)| > \frac{1}{2} \left( \min(n_p,n_q) - |S| \right) \right] \\
    &\leq 
        \sum_{t=0}^{n_p - \frac{\min(n_p,n_q)}{10}}
        \mathbf{Pr}\left[|S|=\min(n_p,n_q)-t,\, |\mathrm{\Gamma}(S)|\leq \min(n_p,n_q) - t/2 \right]  \\
    &= 
        \sum_{t=0}^{n_p - \frac{\min(n_p,n_q)}{10}} \mathbf{Pr} [\exists S\subseteq N_p, T\subseteq M(N_q) :  |S|= \min(n_p,n_q)-t,\,|T|=\min(n_p,n_q) - t/2, \mathrm{\Gamma}(S)\subseteq T]  \\
    &\le 
        \sum_{t=0}^{n_p - \frac{\min(n_p,n_q)}{10}}
        \sum_{\substack{\exists S\subseteq N_p, T\subseteq M(N_q) \\ |S|=\min(n_p,n_q)-t,\, |T|=\min(n_p,n_q)-t/2}} \mathbf{Pr}\left[\mathrm{\Gamma}(S)\subseteq T  \right]  \\
    &\le 
        \sum_{t=0}^{n_p - \frac{\min(n_p,n_q)}{10}} n_p^{\min(n_p,n_q) - t}\cdot n_q^{\min(n_p,n_q)-t/2}  \cdot\left( 1 - c_0 \frac{(\log \min(n_p,n_q))^3}{\min(n_p,n_q)} \right)^{(n_p-\min(n_p,n_q)+t)(n_q-\min(n_p,n_q)+t/2)}  \\
    &\leq \sum_{t=0}^{n_p - \frac{\min(n_p,n_q)}{10}} n_p^{\min(n_p,n_q) - t}\cdot n_q^{\min(n_p,n_q)-t/2} \\
    &\quad \cdot
        \exp\left(  - c_0 \frac{(\log \min(n_p,n_q))^3}{\min(n_p,n_q)} (n_p-\min(n_p,n_q)+t)(n_q-\min(n_p,n_q)+t/2)\right)  \\
    &= 
        O(\min(n_p,n_q)^{-3}).
\end{align*}
Therefore, from these calculation,
we find a desired alternating cycle of length $\ell = O(\log \min(n_p,n_q))$ with a probability of at least $1-O(\min(n_p,n_q)^{-3})$.
\end{proof}


\subsection{Putting All the Pieces Together}\label{sec:final_together_proof}
In this section, we ...

\subsubsection{Lower Bound for $\mathbb{E}[v_p(M(N_p))] - \mathbb{E}[v_p(M(N_q))] $}
Under two conditions~\ref{item:assumption-b} and~\ref{item:assumption-c} of Theorem~\ref{thm:main_theorem_1}, we show that under the matching produced by the round-robin algorithm, any class does not envy other class in expectation if different classes have the same size the minimum number of classes is large. Formally,
\begin{restatable}{lemma}{LemThmDiffBound}\label{lemma:expected_diff}
    Suppose that $\mathcal{D}$ is $(\alpha,\beta)$-PDF-bounded, and two conditions~\ref{item:assumption-b} and~\ref{item:assumption-c} hold. 
    Then, for every pair of classes $p,q\in [k]$, we have %
    \begin{align*}\label{inequality:envy-free-in-expectation}
        &\mathbb{E}[v_p(M(N_p))] - \mathbb{E}[v_p(M(N_q))]  \\
        &\ge
            \left(1- \frac{1}{2\alpha k}\right) \left(n_p- \min(n_p,n_q) \right)+ \left(\frac{\alpha}{\beta} - \frac{n_q+1}{\alpha (m-k)} \right) \frac{\min(n_p,n_q)}{n_q} - O(\min(n_p,n_q)^{-1}).
    \end{align*}
\end{restatable}
We explain implications of Lemma~\ref{lemma:expected_diff}. 
If we have $k > \max(\frac{1}{2\alpha},\frac{\beta}{\alpha^2})$ and $m\ge k \max_{p}(n_p+2)$, then we obtain $1-\frac{1}{2\alpha k} > 0$ and $\frac{\alpha}{\beta} - \frac{n_q+1}{\alpha (m-k)} \ge \frac{\alpha}{\beta} - \frac{1}{\alpha k} > 0$.
Therefore, when $n_p\ge n_q$, there exists a positive constant $c>0$ such that the expected difference is at least $c - O(n_q^{-1})$. 
When $n_p<n_q$, the expected difference is at least $\left(\frac{\alpha}{\beta} - \frac{n_q+1}{\alpha (m-k)}\right)\frac{n_p}{n_q} - O(n_p^{-1})$, where the lower bound is mainly determined by the ratio of $n_p$ to $n_q$. 

Before proceeding to show Lemma~\ref{lemma:expected_diff}, we present the following lemma (Lemma~\ref{lemma:bound_by_pi^2/6}). This lemma is proven by some calculations, and the proof can be found in Appendix~\ref{appendix:bound_by_pi}.
\begin{restatable}{lemma}{LemBoundByPi}\label{lemma:bound_by_pi^2/6}
    \[
        \sum_{r=1}^{\min(n_p,n_q)} \frac{1}{r}\sum_{r'=1}^r \frac{1}{n_q-r'+2}
        \le \frac{\pi^2}{6} + o(1).
    \]
\end{restatable}
We now prove Lemma~\ref{lemma:expected_diff}.
\begin{proof}[Proof of Lemma~\ref{lemma:expected_diff}]
By $($b$)$ in Lemma~\ref{lem:no_heavy_edge},
for every $r=1,2,\ldots,\min(n_p,n_q)$, there is a constant $c>0$ such that
$\mathbf{Pr}\left[ W_q^{(r)} \geq 1- c\frac{(\log \min(n_p,n_q))^4}{\min(n_p,n_q)} \right] $ $= 1-O(\min(n_p,n_q)^{-3}).$
For every $r=1,2,\ldots,\min(n_p,n_q)$, since $W_q^{(r)}\geq 0$, we obtain $\mathbb{E}[W_q^{(r)}] = 1-O(\min(n_p,n_q)^{-1})$.
Thus, combined with Lemma~\ref{claim:expected_diff_q},
we get
\begin{align*}
    \mathbb{E}[X_{pq}] 
    &\le  \min(n_p,n_q) -  \frac{\alpha}{\beta} \sum_{r=1}^{\min(n_p,n_q)} \frac{1}{r}\sum_{r'=1}^r \frac{\mathbb{E}\bigl[W_q^{(r)}\bigr]}{n_q-r'+2} \nonumber \\
    &\le  
        \min(n_p,n_q) -  \frac{\alpha}{\beta} \sum_{r=1}^{\min(n_p,n_q)} \frac{1}{r}\sum_{r'=1}^r \frac{1 - O(\min(n_p,n_q)^{-1})}{n_q-r'+2}  \\
    &\le 
        \min(n_p,n_q) 
        -  \frac{\alpha}{\beta} \sum_{r=1}^{\min(n_p,n_q)} \frac{1}{r}\sum_{r'=1}^r \frac{1}{n_q-r'+2} 
    + 
        \frac{\alpha}{\beta} \sum_{r=1}^{\min(n_p,n_q)} \frac{1}{r}\sum_{r'=1}^r \frac{O(\min(n_p,n_q)^{-1})}{n_q-r'+2}.
\end{align*}
By Lemma~\ref{lemma:bound_by_pi^2/6}, we have $\sum_{r=1}^{\min(n_p,n_q)} \frac{1}{r}\sum_{r'=1}^r \frac{1}{n_q-r'+2} = O(1)$. 
The above inequality and Lemma~\ref{claim:expected_diff_p} implies that
\begin{align*}
    &\mathbb{E}[X_p] - \mathbb{E}[X_{pq}] \\
    &\ge 
        n_p - \frac{1}{\alpha} \sum_{r=1}^{n_p}\frac{1}{r} \sum_{r'=1}^r  \frac{1}{m - r'\cdot k} -\min(n_p,n_q)  +\frac{\alpha}{\beta}\sum_{r=1}^{\min(n_p,n_q)} \frac{1}{r}\sum_{r'=1}^r \frac{1}{n_q-r'+2}- O(\min(n_p,n_q)^{-1}) \\
    &\ge
        n_p-\min(n_p,n_q) -O(\min(n_p,n_q)^{-1}) 
        + \left(\frac{\alpha}{\beta} - \frac{n_q+1}{\alpha (m-k)}\right) \sum_{r=1}^{\min(n_p,n_q)} \frac{1}{r}\sum_{r'=1}^r \frac{1}{n_q-r'+2} \\
    &\quad-
        \frac{1}{\alpha} \sum_{r=\min(n_p,n_q)+1}^{n_p}
        \frac{1}{r}\sum_{r'=1}^r \frac{1}{m-r'\cdot k} \\
    &\geq 
        n_p-\min(n_p,n_q)  - O(\min(n_p,n_q)^{-1}) \\
    &\quad+ 
        \left(\frac{\alpha}{\beta} - \frac{n_q+1}{\alpha (m-k)}\right) \sum_{r=1}^{\min(n_p,n_q)} \frac{1}{r}\sum_{r'=1}^r \frac{1}{n_q} - \frac{1}{\alpha k} \sum_{r=\min(n_p,n_q)+1}^{n_p}
        \frac{1}{r}\sum_{r'=1}^r \frac{1}{2}  \tag{by $m\ge k(n_q+2)$, $\frac{1}{n_q-r'+2}\ge \frac{1}{n_q}$ and $\frac{1}{2}\ge \frac{1}{m-r'\cdot k}$}\\
    &= 
        n_p-\min(n_p,n_q) + \left(\frac{\alpha}{\beta} - \frac{1}{\alpha k}\right)  \frac{\min(n_p,n_q)}{n_q} - \frac{n_p-\min(n_p,n_q)}{2\alpha k}  - O(\min(n_p,n_q)^{-1})\\
    &= 
        \left(1- \frac{1}{2\alpha k}\right) \left(n_p- \min(n_p,n_q) \right)+ \left(\frac{\alpha}{\beta} - \frac{n_q+1}{\alpha (m-k)} \right) \frac{\min(n_p,n_q)}{n_q} 
        - O(\min(n_p,n_q)^{-1}).
\end{align*}
This completes the proof.
\end{proof}


\subsubsection{Proof of Theorem~\ref{thm:main_theorem_1}}

Finally, we prove Theorem~\ref{thm:main_theorem_1}. 
We denote the right hand-side of the inequality in Lemma~\ref{lemma:expected_diff} by $D(n_p,n_q)$.
We first introduce the Efron-Stein inequality and Chebyshev’s inequality as the following.

\begin{lemma}[Efron–Stein inequality~\cite{BoucheronBook,Efron_Stein}]\label{lemma:efron_stein}
    Suppose that $n$ random variables $X_1,X_2,\ldots,X_n$
    are independent. Let $f:\mathbb{R}^{n}\to \mathbb{R}$ be an arbitrary measurable function of $n$ random variables. Let $X = (X_1,X_2,\ldots,X_n)$ and $X^{(i)} = (X_1,X_2,\ldots,X_{i-1}, 0, X_{i+1},\ldots,X_n)$.
    Then, we have
        $\mathrm{Var}[f(X_1,X_2,\ldots,X_n)] \le \sum_{i=1}^n \mathbb{E}[(f(X) - f(X^{(i)}))^2_{+}]$.
    Here, $(x)_{+} = \max(x,0)$.
\end{lemma}
\begin{lemma}[Chebyshev’s inequality]\label{lemma:Chebyshev}
    If $X$ is any random variable, then for any $\varepsilon>0$ we have $\mathbf{Pr}[\,|X - \mathbb{E}[X]| \geq \varepsilon\,] \leq \frac{\mathrm{Var}[X]}{\varepsilon^2}.$
\end{lemma}
Moreover, we use an inequality according to variances.
\begin{lemma}
For arbitrary two random variables $X$ and $Y$,
\begin{equation}\label{eq:variance_inq}
    \mathrm{Var}[X+Y] = \mathrm{Var}[X]+\mathrm{Var}[Y] + \mathrm{Cov}(X,Y) \leq 2(\mathrm{Var}[X]+\mathrm{Var}[Y]),
\end{equation}
where $\mathrm{Cov}(X,Y)$ is the covariance of $X$ and $Y$. 
\end{lemma}


From the condition~\ref{item:assumption-c} in Theorem~\ref{thm:main_theorem_1}, we obtain
\begin{equation}
    \max(n_p,n_q) \le n \le C\cdot (\mathrm{min}_{p\in [k]}n_{p})^{5/4} \le C\cdot \min(n_p,n_q)^{5/4} \quad \forall p,q\in [k].
\end{equation}
From this, we get
\begin{equation}
    n_p\to\infty \quad\text{ as }\quad n\to\infty \quad \forall p\in [k].
\end{equation}

\begin{proof}[Proof of Theorem~\ref{thm:main_theorem_1}]
To show Theorem~\ref{thm:main_theorem_1} by the Chebyshev’s inequality, we will bound the variances of $\mathrm{Var}\left[X_p \right]$ and $\mathrm{Var}\left[X_{pq} \right]$.
\paragraph{Bound for $\mathrm{Var}\left[X_p \right]$.}
First, we investigate the probability that the value of the random variable $X_p$ deviates from its expected value. Let $W(i,j) = u_i(j)$ denote the weight of edge $\{i,j\}$ for each $i \in N$ and $j\in I$. 
Let $X_p(W)$ denote 
denote the total utility that class $p$ obtains from the matching produced by the round-robin algorithm when the input is $(W(i,j))_{i,j}$.
Note that $\mathrm{Var}[X_{p}(W)] = \mathrm{Var}[X_{p}]$.

Consider another weight function. 
Let $\delta_p = c_p \frac{(\log n_p)^2}{n_p}$, and $\overline{W}(i,j) = \max(W(i,j), 1-\delta_p )$ for all $i \in N$ and $j\in I$. 
By the definition, we have $\overline{W}(i,j) \ge 1-\delta_p$ for all $i \in N$ and $j\in I$. 
Let $\overline{W} = (\overline{W}(1,1),\overline{W}(1,2),\ldots,\overline{W}(n,m))$.
We denote by $X_p(\overline{W})$ the total utility attained by class $p$ under the round-robin algorithm when weights of edges are $\{\overline{W}(i,j)\}_{i,j}$ instead of $\{W(i,j)\}_{i,j}$.

We now investigate $\mathrm{Var}\left[ X_p(\overline{W}) \right]$.
Let $\overline{W}^{(i,j)} = (\overline{W}(1,1),\overline{W}(1,2),\ldots,\overline{W}(i,j-1),0, \overline{W}(i,j+1),\ldots,\overline{W}(n,m))$. Similar to $X_p(\overline{W})$, we define $X_p(\overline{W}^{(i,j)})$.
By applying the Efron–Stein inequality (Lemma~\ref{lemma:efron_stein}), we obtain
\begin{align*}
    \mathrm{Var}\left[X_p(\overline{W}) \right] 
    \leq 
        \sum_{i=1}^{n} \sum_{j=1}^{m} \mathbb{E}\Bigl[(X_p(\overline{W}) - X_p(\overline{W}^{(i,j)}))^2_{+}\Bigr].
\end{align*}
By the definition of $\overline{W}$, we have $(X_p(\overline{W}) - X_p(\overline{W}^{(i,j)}))^2_{+} \le 2\delta_p^2$ for all $i\in N$ and $j\in I$. 
We count the number of edges $\{i,j\}$ for which $X_p(\overline{W}^{(i,j)})$ becomes strictly less than $X_p(\overline{W})$ when the weight of edge $\{i,j\}$ is replaced with $0$.
From the proof in Lemma~\ref{lemma:the_length_of_alternating}, for each $p\in [k]$, we have, at least a probability of $1-O(n_p^{-3})$, the length of alternating path is at most $O(\log n_p)$. 
Since the number of edges involved in the matching produced by the algorithm is, with a probability at least $1-O(\min_{p\in [k]}n_p^{-3})$, at most $\sum_{p=1}^k O(\log n_p) n_p$, and
$
    \delta_p^2\sum_{p = 1}^k n_p O(\log n_p)  
    \le 
        \delta_p^2 \sum_{p=1}^k n_p O(\log n)
    =
        \delta_p^2 nO( \log n),
$
we obtain
\[
    \mathbb{E}\Bigl[\sum_{i=1}^{n} \sum_{j=1}^{m} (X_p(\overline{W}) - X_p(\overline{W}^{(i,j)}) )^2_{+}\Bigr] \le \delta_p^2 nO( \log n) + \delta_p^2 n^2 O(\mathrm{min}_{p\in [k]}n_p^{-3})
    \le \delta_p^2 nO( \log n) +  \delta_p^2 n_p^{-1/2}.
\]
Hence, $\mathrm{Var}\bigl[ X_p(\overline{W}) \bigr] \le \delta_p^2 O(n \log n) + O(n_p^{-1}) $.
Furthermore, from (a) in Lemma~\ref{lem:no_heavy_edge}, we have $2 n_p^2 \cdot \mathbf{Pr}\left[X_p(W) \neq X_p(\overline{W}) \right] = 2 n_p^2\cdot O(n_p^{-3}) = O(n_p^{-1})$.

Next, we bound the variance of $X_p(W)$ as follows.
\begin{align*}
    \mathrm{Var}[X_p(W)] 
    &= 
        \mathrm{Var}\left[X_p(W)\cdot \mathbbm{1}[X_p(W)=X_p(\overline{W})] + X_p(W)\cdot \mathbbm{1}[X_p(W) \neq X_p(\overline{W})]\right]  \\
    &\leq 
        2 \mathrm{Var}\left[\overline{X}_p \right] 
        + 2 \mathrm{Var}\left[ X_p(W)\cdot \mathbbm{1}[X_p(W) \neq X_p(\overline{W})] \right] \tag{by Inequality~\ref{eq:variance_inq}}\\
    &\leq 
        2 \mathrm{Var}\left[\overline{X}_p \right] 
        + 2 n_p^2 \cdot\mathrm{Var}\left[\mathbbm{1}[X_p(W) \neq X_p(\overline{W})] \right] \tag{by $X_p \leq n_p$} \\
    &\leq 
        2 \mathrm{Var}\left[ X_p(\overline{W})\right] 
        + 2 n_p^2 \cdot\mathbf{Pr}\left[X_p(W) \neq X_p(\overline{W}) \right] \tag{by $\mathrm{Var}\left[\mathbbm{1}(\mathcal{A}) \right] = \mathbf{Pr}\left[\mathcal{A}\right] - \mathbf{Pr}\left[\mathcal{A}\right]^2 \leq \mathbf{Pr}\left[\mathcal{A}\right]$ for an event $\mathcal{A}$}\\
    &\le 
        2 \delta_p^2 nO( \log n) + O(n_p^{-1}).
\end{align*}
Finally, from the above inequality and the Chebyshev's inequality, we achieve
\begin{equation}\label{eq:D(n_p,n_q)_1}
    \mathbf{Pr}\left[X_p - \mathbb{E}[X_p] < -\frac{1}{2}D(n_p,n_q)\right] 
    \leq 
        \frac{4 \mathrm{Var}\left[X_p \right]}{D(n_p,n_q)^2}
    \leq 
        \frac{8\delta_p^2 n O(\log n)  + O(n_p^{-1})}{D(n_p,n_q)^2}.
\end{equation}

\paragraph{Bound for $\mathrm{Var}\left[ X_{pq} \right]$}

Subsequently, we investigate the concentration around the expected value of the random variable $X_{pq}$. Similarly, let us define $\delta_{pq} = c_{q} \frac{(\log \min(n_p,n_q))^4}{\min(n_p,n_q)}$.
We define $\widetilde{W}(i,j) = \max(W(i,j), 1-\delta_{pq} )$ for all $i \in N$ and $j\in I$.
Let $X_{pq}(W)$ be the function that receives all utilities and returns the assignment valuation of class $p$ to class $q$ when the input is $W$.
We also define $\widetilde{W}^{(i,j)} = (\widetilde{W}(1,1),\widetilde{W}(1,2),\ldots,\widetilde{W}(i,j-1),0, \widetilde{W}(i,j+1),\ldots,\widetilde{W}(n,m))$, and
define the random variable of the assignment valuation for class $p$ toward the bundle of class $q$ induced by the output of the round-robin using $\{\widetilde{W}^{(i,j)}\}_{i,j}$ as 
$X_{pq}(\widetilde{W}^{(i,j)})$.
Similarly to the bound for $\mathrm{Var}[X_p(W)]$, we obtain 
\[
    \mathrm{Var}[X_{pq}(W)] 
    \leq 2 \mathrm{Var}\left[X_{pq}(\widetilde{W}^{(i,j)}) \right] 
    + 2 \min(n_p,n_q)^2\cdot \mathbf{Pr}\left[X_{pq}\neq\overline{X}_{pq} \right].
\]
From (b) in Lemma~\ref{lem:no_heavy_edge}, we obtain $\min(n_p,n_q)^2 \cdot\mathbf{Pr}\left[X_{pq}\neq\overline{X}_{pq} \right]= O(\min(n_p,n_q)^{-1})$.
By applying the Efron-Stein inequality, we obtain
\[
    \mathrm{Var}[X_{pq}(W)] 
    \leq 
        2 \delta_{pq}^2 nO( \log n) + O(\min(n_p,n_q)^{-1}).
\]
Thus, from the Chebyshev's inequality,
\begin{equation}\label{eq:D(n_p,n_q)_2}
    \mathbf{Pr}\left[X_{pq} - \mathbb{E}[X_{pq}] 
    > 
        \frac{1}{2}D(n_p,n_q)\right] 
    \leq 
        \frac{4 \mathrm{Var}\left[X_{pq} \right]}{D(n_p,n_q)^2}
    \leq 
        \frac{8 \delta_{pq}^2 nO( \log n) + O(\min(n_p,n_q)^{-1}) }{D(n_p,n_q)^2}.
\end{equation}


We now have 
\[
    \delta_{p}^2 + \delta_{pq}^2 
    = 
        c_p n_p^{-2}(\log n_p)^4 + c_q \min(n_p,n_q)^{-2} (\log \min(n_p,n_q))^{8} \le 2 c_q  \min(n_p,n_q)^{-2} (\log n)^{8}.
\]

\paragraph{The probability that class $p$ envies class $q$.}
Finally, we bound the probability that class $p$ envies class $q$.
From Lemma~\ref{lemma:expected_diff}, we have $\mathbb{E}[X_p] - \mathbb{E}[X_{pq}] \geq D(n_p,n_q)$.
If class $p$ envies class $q$, then $X_p < \mathbb{E}[X_p] - \frac{1}{2}D(n_p,n_q)$, or $X_{pq} > \mathbb{E}[X_{pq}] +\frac{1}{2} D(n_p,n_q)$ must hold. 
Hence, from~\eqref{eq:D(n_p,n_q)_1} and~\eqref{eq:D(n_p,n_q)_2}, we obtain
\begin{align*}
    \mathbf{Pr}\left[\text{Class $p$ envies class $q$} \right] 
    &=
        \mathbf{Pr}\left[X_p < X_{pq} \right] \\
    &\leq 
        \mathbf{Pr}\left[ X_p < \mathbb{E}[X_p] - \frac{1}{2}D(n_p,n_q)\right] +\mathbf{Pr}\left[X_{pq} > \mathbb{E}[X_{pq}] + \frac{1}{2}D(n_p,n_q)\right] \\
    &\le 
        \frac{8 \delta_p^2 nO( \log n) + O(n_p^{-1})}{D(n_p,n_q)^2} +\frac{8 \delta_{pq}^2 nO( \log n) + O(\min(n_p,n_q)^{-1})}{D(n_p,n_q)^2} \\
    &\le 
        \left( 8 (\delta_p^2 + \delta_{pq}^2) nO( \log n) + O(\min(n_p,n_q)^{-1})\right) \cdot D(n_p,n_q)^{-2} \\
    &\le 
        \left( 16 c_q\cdot \min(n_p,n_q)^{-2} nO((\log n)^9) + O(\min(n_p,n_q)^{-1})\right) \cdot D(n_p,n_q)^{-2},
\end{align*}
where we use $\delta_{p}^2 + \delta_{pq}^2 \le 2 c_q  \min(n_p,n_q)^{-2} (\log n)^{8}$ for the last inequality.

Now, we consider the two cases: $n_p\ge n_q$ or $n_p < n_q$.
When $n_p\ge n_q$, there exist a constant $c>0$ such that $D(n_p,n_q) \ge c - O(n_q^{-1})$. Then, we get
\begin{align*}
    \mathbf{Pr}\left[\text{Class $p$ envies class $q$} \right]
    &\le 
        (16 c_q\cdot n_q^{-2} nO((\log n)^9) + O(n_q^{-1}) ) \cdot D(n_p,n_q)^{-2} \\
    &\le 
        (16 c_q\cdot n_q^{-2} nO((\log n)^9) + O(n_q^{-1}) ) \cdot (c - O(n_q^{-1}))^{-2} \\
    &= 
        16 c_q c^{-2}\cdot n_q^{-2} nO((\log n)^9) + O(n_q^{-1})  \\
    &\le 
        16 c_q c^{-2} C\cdot n_q^{-3/4} O((\log n)^9) + O(n_q^{-1})  \tag{by $n \le C\cdot n_q^{5/4}$}\\
    &= 
        \tilde{O}(n_q^{-3/4}) \\
    &= 
        \tilde{O}(n^{-3/5}),
\end{align*}
where $\tilde{O}$ is the big-O notation that ignores logarithmic factors.
When $n_p < n_q$, we have $D(n_p,n_q) \ge (\frac{\alpha}{\beta}-\frac{1}{\alpha k})\frac{n_p}{n_q} - O(n_p^{-1})$. Then,
we obtain
\begin{align*}
    D(n_p,n_q)^{-2} 
    &\le 
        \left(\left(\frac{\alpha}{\beta}-\frac{1}{\alpha k}\right)\frac{n_p}{n_q} - O(n_p^{-1})\right)^{-2} \\
    &\le 
        \left(\left(\frac{\alpha}{\beta}-\frac{1}{\alpha k}\right)C^{-1}n_p^{-1/4} - O(n_p^{-1})\right)^{-2} \tag{Using $n_q \leq C\cdot n_p^{5/4}$}\\
    &=  
        O(n_p^{1/2}).
\end{align*}
Thus, we obtain
\begin{align*}
    \mathbf{Pr}\left[\text{Class $p$ envies class $q$} \right] 
    &\le 
        (16 c_q\cdot n_p^{-2} nO((\log n)^9) + O(n_p^{-1}) ) \cdot D(n_p,n_q)^{-2} \\
    &= 
        (16 c_q\cdot n_p^{-2} nO((\log n)^9) + O(n_p^{-1}) ) \cdot O(n_p^{1/2}) \\
    &= 
        (16 c_q C\cdot n_p^{-3/4} O((\log n)^9) + O(n_p^{-1}) ) \cdot O(n_p^{1/2}) \tag{by $n \leq C\cdot n_p^{5/4}$}\\
    &= 
        \tilde{O}(n_p^{-1/4}) \\
    &= 
        \tilde{O}(n^{-1/5}).
\end{align*}

Therefore, we establish that the probability that class $p$ envies class $q$ approaches $0$ as $n\to \infty$ for all $p,q\in [k]$.
Finally, from condition~\ref{item:assumption-a}, we achieve
\begin{align*}
    \mathbf{Pr}\left[\text{Matching $M$ is not class envy-free} \right] 
    &\le 
        \sum_{p,q\in [k]} \mathbf{Pr}\left[\text{Class $p$ envies class $q$} \right] \\
    &\le 
        k^2\cdot \tilde{O}(n^{-1/5}) \\
    &= 
        o(1).
\end{align*}
This concludes the proof of Theorem~\ref{thm:main_theorem_1}.

\end{proof}


\section{Conclusion}
This paper addressed the problem of achieving fairness in matching across different classes and suggests several open problems for future research. In Theorem~\ref{thm:main_theorem_1}, we made several assumptions. In particular, while \cite{Dickerson2014} and \cite{ManurangsiSuksompong2021}, as well as our Theorem~\ref{thm:maxmatching}, provide asymptotic results when the number of items approaches infinity, Theorem~\ref{thm:main_theorem_1} assumes the number of agents approaches infinity. This assumption is necessary for our analysis because the assignment valuation of each class is bounded by the number of agents in that class. However, it remains unclear whether the round-robin algorithm produces asymptotically a class envy-free matching when these assumptions do not hold. We leave this as an interesting open question for future work. Also, it would be interesting to explore the case of asymmetric agents, where each agent has a different probability distribution for their utilities. 


\section*{Acknowledgments}
This work was partially supported by JST FOREST Grant Numbers JPMJPR20C1 and by JST ERATO under grant number JPMJER2301. We thank the anonymous AAMAS 2025 for their valuable comments.







\bibliographystyle{plain} 
\bibliography{full_version}

\appendix

\section{Omitted Material From Section~\ref{sec:Preliminaries}}\label{app:nesting_lemma}


\begin{proof}[Proof of Lemma~\ref{lemma:uniqueness_of_maximum_weight_matching}]
Let $w(e)$ denote the weight of each edge $e$ in $H$.
Let $M_1$ and $M_2$ be two distinct matchings in $H$. 
For an edge $e \in M_i\setminus M_j$ with $i \neq j$ and $i,j \in \{1,2\}$, we define $s(e)$ by
$
    s(e) 
    =  w(M_i) - \sum_{e' \in M_j \setminus \{e\}} w(e'),
$
where $w(M_i)= \sum_{e'\in M_i}w(e')$.
Here, the random variable $s(e)$ is independent of $w(e)$. Thus, given $s(e)$, we have $\mathbf{Pr}[w(e)=s(e)]=0$ due to the non-atomicity of the distribution.
If $w(M_1)=w(M_2)$, then it follows that $w(e)=s(e)$ by the definition of $s(e)$. Let $\mathcal{E}$ denote the event that there exist two distinct matchings of the same total weight. 
Now, we have 
\[
\mathbf{Pr}[\mathcal{E}] \leq \sum_{M_1,M_2 \in \mathcal{M},\,  M_1\neq M_2} \ \sum_{e\in M_1\Delta M_2}\mathbf{Pr}[w(e)=s(e)] = 0,
\] 
where $\mathcal{M}$ is the set of matchings in $H$. 
This concludes that no two distinct matchings have the same total weight almost surely.
\end{proof}


\begin{proof}[Proof of Lemma~\ref{lemma:nesting_lemma}]
In this proof, let  $M^r$ denote the maximum-weight matching with $r$ edges in $H[A',B']$. Fix $r>1$. 
For a subset $E' \subseteq E$, let $w(E')$ denote the total weight of edges in $E'$.
From Lemma~\ref{lemma:uniqueness_of_maximum_weight_matching}, no two distinct matchings have the same total weight in $H[A',B']$ almost surely. 
Let $E_0=M^{r-1} \Delta M^r$.
Since all vertices have degree of at most two, $E_0$ consists of paths or cycles.
We claim that $E_0$ consists of a single path.  
Suppose towards a contradiction that $E_0$ is not composed of a single path. Then $|M^r \setminus M^{r-1}|\geq 1$ and $|M^{r-1} \setminus M^{r}|\geq 1$. Moreover, there is a subset $E_1 \subseteq E_0$ such that $E_1$ includes an equal number of edges from each $M^{r-1}$ and $M^r$, and $E_1 \Delta M^{r-1}$ and $E_1 \Delta M^{r}$ are matchings.

To see this, let $C_1,\ldots,C_{\ell}$ be the connected components of $E_0$. There exists at least one $C_j \in \{C_1,\ldots,C_{\ell}\}$ that is a single path with $|C_j \cap M^r| = |C_j \cap M^{r-1}|+1$. By our assumption, $\ell \geq 2$. If there is a connected component $C_{j'}$ that includes an equal number of edges from each $M^{r-1}$ and $M^r$, then setting $E_1=C_{j'}$ satisfies the desired properties. Suppose that every connected component $C_{j'}$ contains a different number of edges from $M^{r-1}$ and $M^r$, which means $||C_{j'} \cap M^r| - |C_{j'} \cap M^{r-1}||=1$. This together with the fact that $\sum_{j' \in [\ell]\setminus \{j\}} |C_{j'} \cap M^r|= \sum_{j' \in [\ell]\setminus \{j\}} |C_{j'} \cap M^{r-1}|$
implies that there must exist a pair of connected components with the same total number of edges from $M^{r-1}$ and from $M^r$. We can select such a pair as $E_1$. 

Note that $E_1 \Delta M^{r-1}$ is a matching of size $r-1$ in $H[A',B']$, and $E_1 \Delta M^{r}$ is a matching of size $r$ in $H[A',B']$. Thus, by Lemma~\ref{lemma:uniqueness_of_maximum_weight_matching}, 
\[
    w(E_1 \Delta M^{r-1}) \neq w(E_1 \Delta M^{r}).
\]
Now, we have
\begin{align*}
    w(E_1 \Delta M^{r-1}) + w(E_1 \Delta M^{r}) 
    &= w(M^{r-1} \setminus E_1) + w(M^r \cap E_1) + w(M^{r} \setminus E_1) + w(M^{r-1} \cap E_1) \\
    &= w(M^{r-1}) + w(M^{r}).
\end{align*}
From these, it follows that either $w(E_1 \Delta M^{r-1}) > w(M^{r-1}) $ or $w(E_1 \Delta M^{r}) > w(M^{r}) $. 
This contradicts the fact that $M^{r-1}$ and $M^r$ are the maximum-weight matching with $r-1$ and $r$ edges, respectively, in $H[A',B']$.
Therefore, $E_0$ must consist of a single path. 
When creating $M^r$ from $M^{r-1}$ along the single alternating path, every vertex that appears in $M^{r-1}$ continues to appear in $M^r$.
\end{proof}


\section{Omitted Material From Section~\ref{sec:mainresults}}\label{app:for_contributions}

\begin{restatable}{proposition}{CompleteClassEnvy}\label{complete_class_envy-free}
    Suppose that that $\mathcal{D}$ is non-atomic. Let $\mathrm{e}$ be Napier's constant and $\varepsilon>0$ be any constant. If $m\geq (\mathrm{e}+\varepsilon)\cdot \sum_{p\in [k]} n_p$, then the probability that a class envy-free matching under which every agent obtains exactly one item exists approaches $1$ as $m \to \infty$. 
\end{restatable}

\renewcommand{\proofname}{Proof}
\begin{proof}
    We say that a matching $M$ is \emph{envy-free} if we have $u_i(j) \geq u_i(j')$ for every pair of agents $i,i' \in N$ where $j$ (resp. $j'$) is an item matched to $i$ (resp. $i'$) under $M$. 
    The proof of Theorem 6.1 in~\cite{ManurangsiSuksompong2021} implies that, 
    when $m\geq (\mathrm{e}+\varepsilon)\cdot n = (\mathrm{e}+\varepsilon)\cdot \sum_{p\in [k]} n_p$, a greedy assignment algorithm for an envy-free matching (Algorithm $6.1$ in~\cite{ManurangsiSuksompong2021}) computes an envy-free matching $M$ with a probability that converges to $1$ as $n\to \infty$.
    Let $\mathcal{M}(N_p,M(N_q))$ denote the set of all maximum-weight matchings between $N_p$ and $M(N_q)$.
    Observe that since $M$ is envy-free, we have that, for every pair $p,q\in [k]$ of classes and every matching $M'\in \mathcal{M}(N_p,M(N_q))$, 
    \[
    u_p(M)=\sum_{\{i,j\} \in M,\, i\in N_p}u_i(j) \geq 
    \sum_{\{i,j\}\in M'}u_i(j).
    \]
    Thus, $u_p(M) \geq v_p(M(N_q))$ for every pair $p,q\in [k]$ of classes.     
    This concludes the proof. 
\end{proof}


\section{Proof of Proposition~\ref{proposition:1_2_ef1_non_waste}}

\begin{proof}[Proof of Proposition~\ref{proposition:1_2_ef1_non_waste}]
A matching $M$ is \emph{marginal class envy-free up to one item (MCEF1)} if for every pair of classes $p, q \in [k]$, class $p$ does not envy class $q$, or there exists an item $j\in M(N_q)$ such that $u_p(M) \geq v_p(M(N_p) \cup M(N_q)\setminus \{j\}) - u_p(M)$.
The resulting matching $M$ is MCEF1 by Corollary~3.2 in~\cite{Montanari2024} and $1/2$-CEF1 by Theorem 3.8 in~\cite{Amanatidis2023}. 

Next, we prove that matching $M$ is non-wasteful.
It is immediate that the matching satisfies the condition~\ref{item:consition-a-of-non-wasteful} in Definition~\ref{def:nowasteful} by construction of the algorithm.
To show that $M$ satisfies the condition~\ref{item:consition-b-of-non-wasteful} in Definition~\ref{def:nowasteful}, we prove a stronger statement that no class is allocated to an item whose removal does not change their total utility. 
We use the fact that for assignment valuations, the greedy algorithm that constructs a size-$r$ subset by sequentially adding an item with the highest marginal utility is guaranteed to compute an optimal set of items that maximizes the valuation across all size-$r$ subsets (See Theorem~17.2 in~\cite{Schrijver}). 
This implies that for each round $r$ and each class $p$, we have
\begin{equation}\label{eq:greedy}
    M^{r}(N_p) \in \mathrm{argmax} \bigl\{\, v_p(I')\, \bigm|\, \, |I'|=r ~\land~I' \subseteq M^{r-1}(N_p) \cup I_p^{r}\,\bigr\}.
\end{equation}
Suppose towards a contradiction that there exist class $p$ and item $j \in M(N_p)$ such that $v_p(M(N_p)) - v_p(M(N_p) \setminus \{j\})=0$. Then, $v_p(M(N_p) \setminus \{j\})=v_p(M(N_p))$. Since each item $j_p^r$ added to the bundle of class $p$ has a positive contribution, we have
$$
    v_p(M(N_p))>v_p(M^{r^*-1}(N_p)),
$$
where $r^*$ denotes the last round when class $p$ obtains some item. However, this implies $v_p(M(N_p) \setminus \{j\})>v_p(M^{r^*-1}(N_p))$, which contradicts~\eqref{eq:greedy} since $M(N_p) \setminus \{j\}$ has the same cardinality as $M^{r^*-1}(N_p)$ and yields a higher value for class $p$. 
\end{proof}


\section{
Proof of Lemma~\ref{lemma:bound_by_pi^2/6}
}\label{appendix:bound_by_pi}

\begin{proof}[Proof of Lemma~\ref{lemma:bound_by_pi^2/6}]
    From $\sum_{i=1}^n \frac{1}{i} = \log n + \gamma + \frac{1}{2n} + O(n^{-2})$, where $\gamma$ is a constant~\cite{young1991euler},
    we have
    \begin{align*}
        &\sum_{r=1}^{\min(n_p,n_q)} \frac{1}{r}\sum_{r'=1}^r \frac{1}{n_q-r'+2} \\
        &= 
            \sum_{j=1}^{\min(n_p,n_q)} \frac{1}{n_q-j+2} \sum_{r=j}^{\min(n_p,n_q)} \frac{1}{r} \\
        &= 
            \sum_{j=1}^{\min(n_p,n_q)} \frac{1}{n_q-j+2} \left(\log \min(n_p,n_q) + \frac{1}{2\min(n_p,n_q)} - \log j - \frac{1}{2j} + O(n^{-2})\right) \\
        &= 
            \sum_{j=1}^{\min(n_p,n_q)} \frac{1}{n_q-j+2} \log\left(\frac{\min(n_p,n_q)}{j}\right) + o(1) \\
        &= 
            \sum_{j=n_q - \min(n_p,n_q)}^{n_q-1} \frac{1}{i+2} \log\left(\frac{\min(n_p,n_q)}{n_q - i}\right) + o(1) \\
        &= 
            -\int_{x = n_q- \min(n_p,n_q)}^{n_q-1} \frac{1}{x} \log\left(1- \frac{x-n_q+\min(n_p,n_q)}{\min(n_p,n_q)}\right) \mathrm{d}x + o(1) \\
        &= 
            -\int_{y=0}^{\min(n_p,n_q)} \frac{1}{y + n_q- \min(n_p,n_q)} \log\left(1- \frac{y}{\min(n_p,n_q)}\right) \mathrm{d}y + o(1) \\
        &=
            \sum_{s=1}^{\infty}\int_{y=0}^{\min(n_p,n_q)} \frac{1}{y + n_q- \min(n_p,n_q)} \frac{y^s}{s\cdot \min(n_p,n_q)^s} \mathrm{d}y + o(1) \\
        &\le 
            \sum_{s=1}^{\infty}\int_{y=0}^{\min(n_p,n_q)} \frac{y^{s-1}}{s\cdot \min(n_p,n_q)^s} \mathrm{d}y + o(1) \\
        &= 
            \sum_{s=1}^{\infty}\frac{1}{s^2} + o(1) \\
        &=
            \frac{\pi^2}{6} + o(1).
    \end{align*}
Here, for the fourth equation, we use the fact the remaining terms are all $o(1)$ when transforming the sum into an integral, and for the last equation, we adapt the well-known result $\sum_{s=1}^{\infty}\frac{1}{s^2} = \frac{\pi^2}{6}$. See~\cite{hofbauer2002simple}.
\end{proof}





\end{document}